\documentclass[a4paper,UKenglish,cleveref, autoref, thm-restate,authorcolumns]{lipics-v2019}

\usepackage{mathtools}
\usepackage{amsfonts}
\usepackage{comment}
\usepackage{algorithm}
\usepackage[]{algpseudocode}
\usepackage{xspace}
\usepackage{todonotes}
\usepackage{caption}
\usepackage{ulem}
\usepackage{multirow}
\usepackage{subcaption}
\usepackage{ifthen}
\allowdisplaybreaks
\nolinenumbers
\title{Kruskal-based approximation algorithm for the multi-level Steiner tree problem} 

\author{Reyan~Ahmed}{University of Arizona, Tucson, United States}{abureyanahmed@email.arizona.edu}{}{}
\author{Faryad~Darabi~Sahneh}{University of Arizona, Tucson, United States}{faryad@email.arizona.edu}{}{}
\author{Keaton~Hamm}{University of Arizona, Tucson, United States}{hamm@email.arizona.edu}{}{}
\author{Stephen~Kobourov}{University of Arizona, Tucson, United States}{kobourov@cs.arizona.edu}{}{}
\author{Richard~Spence}{University of Arizona, Tucson, United States}{rcspence@email.arizona.edu}{}{}

\authorrunning{R.\ Ahmed, et al.} 

\Copyright{Reyan Ahmed, Faryad Darabi Sahneh, Keaton Hamm, Stephen Kobourov, and Richard Spence} 



\category{} 

\relatedversion{} 

\supplement{All algorithms, implementations, the ILP solver, experimental data and analysis are available on Github at \url{ https://github.com/abureyanahmed/Kruskal_based_approximation}.
}

\newcommand\OPT{\ensuremath{\mathrm{OPT}}\xspace}
\newcommand\QoSU{\ensuremath{\mathrm{C_1}}\xspace}
\newcommand\QoSNU{\ensuremath{\mathrm{C_{2a}}}\xspace}
\newcommand\Kruskalu{\ensuremath{\mathrm{KruskalMLST}}\xspace}
\newcommand\Kruskalnu{\ensuremath{\mathrm{GreedyMLST}}\xspace}
\newcommand\CMP{\ensuremath{\mathrm{CMP}}\xspace}

\begin{document}

\maketitle

\begin{abstract}
    We study the multi-level Steiner tree problem: a generalization of the Steiner tree problem in graphs where terminals $T$ require varying priority, level, or quality of service. In this problem, we seek to find a minimum cost tree containing edges of varying rates such that any two terminals $u$, $v$ with priorities $P(u)$, $P(v)$ are connected using edges of rate $\min\{P(u),P(v)\}$ or better. The case where edge costs are proportional to their rate is approximable to within a constant factor of the optimal solution. For the more general case of non-proportional costs, 
    this problem is hard to approximate
    with ratio $c \log \log n$, where $n$ is the number of vertices in the graph. 
    A simple greedy algorithm by Charikar et al., however,  provides a $\min\{2(\ln |T|+1), \ell \rho\}$-approximation in this setting, where $\rho$ is an approximation ratio for a heuristic solver for the Steiner tree problem and $\ell$ is the number of priorities or levels (Byrka et al. give a Steiner tree algorithm with $\rho\approx 1.39$, for example).
    
    In this paper, we describe a natural generalization to the multi-level case of the classical (single-level) Steiner tree approximation algorithm based on Kruskal's  minimum spanning tree algorithm. 
    We prove that this algorithm achieves an approximation ratio at least as good as Charikar et al., and experimentally performs better with respect to the optimum solution. We develop an integer linear programming formulation to compute an exact solution for the multi-level Steiner tree problem with non-proportional edge costs and use it to evaluate the performance of our algorithm on both random graphs and multi-level instances derived from SteinLib.
    
\end{abstract}


\newpage
\section{Introduction}
We study the following generalization of the Steiner tree problem where terminals have priorities, levels, or quality of service (QoS) requirements. Variants of this problem are known in the literature under different names including multi-level network design (MLND), quality-of-service multicast tree (QoSMT) ~\cite{Charikar2004ToN}, quality-of-service Steiner tree~\cite{Xue2001, Karpinski2005}, and Priority Steiner Tree~\cite{Chuzhoy2008}. Motivated by multi-level graph visualization, we refer to this problem as the \emph{multi-level Steiner tree} problem.
\begin{definition}[Multi-level Steiner tree (MLST)] \label{def:mlst}
Let $G=(V,E)$ be a connected graph, and $T\subseteq V$ be a subset of terminals. Each terminal $t \in T$ has a priority $P(t) \in \{1,2,\ldots,\ell\}$. A multi-level Steiner tree (MLST) is a tree $G'$ with edge rates $y(e) \in \{1,2,\ldots,\ell\}$ such that for any two terminals $u$, $v \in T$, the $u$--$v$ path in $G'$ uses edges of rate greater than or equal to $\min\{P(u), P(v)\}$.
\end{definition}
We use 1 for the lowest priority and $\ell$ for the highest, and assume without loss of generality that there exists $v \in V$ such that $P(v) = \ell$. If $\ell=1$, then Definition~\ref{def:mlst} reduces to the definition of Steiner tree.

The cost of an MLST $G'$ is defined as the sum of the edge costs in $G'$ at their respective rates. Specifically, for $1 \le i \le \ell$, we denote by $c_i(e)$ the cost of including edge $e$ with rate $i$, in which the cost of an MLST is $\sum_{e \in E(G')}c_{y(e)}(e)$. Naturally, an edge with a higher rate should be more costly, so we assume that $c_1(e) \le c_2(e) \le \ldots \le c_{\ell}(e)$ for all $e \in E$. The MLST problem is to compute a MLST with minimum cost.

We note that equivalent formulations~\cite{Charikar2004ToN, Chuzhoy2008} include a root (or source) vertex $r \in V$ in which the problem is to compute a tree rooted at $r$ such that the path from $r$ to every terminal $t \in T$ uses edges of rate at least as good as $P(t)$. One can observe that Definition~\ref{def:mlst} is equivalent to this formulation as we can fix the root to be any terminal $r\in T$ such that $P(r)=\ell$. In an optimized multilevel Steiner tree, each path from the root to any terminal uses non-increasing edge rates. Note that this becomes relevant for the discussion of the exact value of the approximation given by our algorithm and the state-of-the-art algorithm~\cite{Charikar2004ToN}. We use the phrase ``multi-level'' since a tree $G'$ with a root having top priority and edge rates $y(\cdot)$ induces a sequence of $\ell$ nested Steiner trees, where the tree induced by $\{e\in E: y(e) \ge i\}$ is a Steiner tree over terminals $T_i=\{t\in T: P(t) \ge i\}$ for $1 \le i \le \ell$.

We distinguish the special case with proportional costs, where the cost of an edge is equal to its rate multiplied by some ``base cost'' (e.g., $c_1(e)$). This is similar to the \emph{rate model} in~\cite{Charikar2004ToN} as well as the setup in~\cite{Karpinski2005}.
\begin{definition}\label{def:proportional}
An instance of the MLST problem has \emph{proportional costs} if $c_i(e) = ic_1(e)$ for all $e \in E$ and for all $i \in \{1,2,\ldots,\ell\}$. Otherwise, the instance has \emph{non-proportional costs}.
\end{definition}

For $u,v \in T$, we define $\sigma(u,v)$ to equal the cost of a minimum cost $u$--$v$ path in $G$ using edges of rate $\min\{P(u), P(v)\}$. 
In other words, $\sigma(u,v)$ represents the minimum possible cost of connecting $u$ and $v$ using edges of the appropriate rate. Note that $\sigma$ is symmetric, but does not satisfy the triangle inequality, and is not a metric. Lastly, we denote by $H_k$ the $k^{\text{th}}$ harmonic number given by $H_k = 1 + \frac{1}{2} + \ldots + \frac{1}{k}$.

\subsection{Related work}
The Steiner tree (ST) problem admits a simple $2\left(1 - \frac{1}{|T|}\right)$-approximation (see Section~\ref{subsec:preliminaries-mst}). Currently, the best known approximation ratio is $\rho = \ln 4 + \varepsilon \approx 1.39$ by Byrka et al.~\cite{Byrka2013}. It is NP-hard to approximate the ST problem with ratio better than $\frac{96}{95} \approx 1.01$~\cite{Chlebnik2008}.

In~\cite{MLST2018}, simple top-down and bottom-up approaches are considered for the MLST problem with proportional costs. In the top-down approach, a Steiner tree is computed over terminals $\{v \in T: P(v) = \ell\}$. For $i = \ell-1$, \ldots, 1, the Steiner tree over terminals $\{v \in T: P(v) = i+1\}$ is contracted into a single vertex, and a Steiner tree is computed over terminals with $P(v) = i$. In the bottom-up approach, a Steiner tree is computed over all terminals, which induces a feasible solution by setting the rate of all edges to $\ell$.  These approaches are $(\frac{\ell+1}{2})\rho$- and $\ell \rho$-approximations, respectively~\cite{MLST2018} (moreover, these bounds are tight). It is worth noting that the bottom-up approach can perform arbitrarily poorly in the non-proportional setting. 

If edge costs are proportional, Charikar et al.~\cite{Charikar2004ToN} give a simple $4\rho$-approximation algorithm (which we later denote by \QoSU) by rounding the vertex priorities up to the nearest power of 2, then computing a $\rho$--approximate Steiner tree for the terminals at each rounded-up priority. They then give an $e\rho$-approximation for the same problem (using the $1.55$-approximation algorithm to compute Steiner tree provided by~\cite{Robins2005}, hence obtaining $e\rho\approx 4.213$). Karpinski et al.~\cite{Karpinski2005} tighten the analysis from~\cite{Charikar2004ToN} to show that this problem admits a 3.802-approximation with an unbounded number of priorities. Ahmed et al.~\cite{MLST2018} generalize the above techniques by considering a composite heuristic which computes Steiner trees over a subset of the priorities, and show that this achieves a $2.351\rho \approx 3.268$-approximation for $\ell \le 100$.  They provide experimental comparisons of the simple top-down, bottom-up, $4\rho$-approximation of Charikar et al.~\cite{Charikar2004ToN}, and a generalized composite algorithm. The experiments in~\cite{MLST2018} show that the bottom-up approach typically provides the worst performance while the composite algorithm typically performs the best, and these results match the theoretical guarantees. 

For non-proportional costs, which is the more general setting, Charikar et al.~\cite{Charikar2004ToN} give a $\min\{2(\ln |T|+1), \ell \rho\}$-approximation for QoSMT, consisting of taking the better solution returned by two sub-algorithms (which we denote by \QoSNU and $\mathrm{C_{2b}}$, Section~\ref{section:charikar}). On the other hand, Chuzhoy et al.~\cite{Chuzhoy2008} show that PST cannot be approximated with ratio better than $\Omega(\log \log n)$ in polynomial time unless NP$\,\subseteq\,$DTIME$(n^{O(\log\log\log n)})$. However, the problem setup for PST~\cite{Chuzhoy2008} is slightly more specific; each edge has a single cost $c_e$ and a Quality of Service (priority) given as input, and a solution consists of a tree such that the path from the root to each terminal $t$ uses edges of QoS at least as good as $P(t)$.  


\subsection{Our contributions}
In this paper, we propose approximation algorithms for the MLST problem based on Kruskal's and Prim's algorithms for computing a minimum spanning tree (MST).  We show that the Kruskal-based algorithm is a $2 \ln |T|$-approximation even for non-proportional costs, matching the state-of-the-art algorithms.  An interesting feature of this algorithm is that for the single level case, it reduces to the standard Kruskal approximation to the Steiner tree problem, which is not the case of other state-of-the-art algorithms for MLST. We also show that, somewhat surprisingly, a natural approach based on Prim's algorithm can perform rather poorly.  We then describe an integer linear program (ILP) to compute exact solutions to the MLST problem given non-proportional edge costs and evaluate the approximation ratios of the proposed approximation algorithms experimentally. Specifically, we provide an experimental comparison between the algorithm of Charikar et al.~\cite{Charikar2004ToN} and our Kruskal-based algorithm, in which the latter performs better with respect to the optimum a majority of the time in both proportional and non-proportional settings.  Experiments are performed on random graphs from various generators as well as instances of the MLST problem derived from the SteinLib library~\cite{KMV00} of hard ST instances. Finally, we describe a class of graphs for which the Kruskal-based algorithm always performs significantly better than that by Charikar et al.~\cite{Charikar2004ToN}. 


\section{Preliminaries}

In this section, we review some existing approximation algorithms that are pivotal for the subsequent developments in this paper.

\subsection{Kruskal- and Prim-based approximations for the ST problem}\label{subsec:preliminaries-mst}

A well-known $2\left(1 - \frac{1}{|T|}\right)$-approximation algorithm for the ST problem first constructs the \emph{metric closure graph} $\tilde{G}$ over $T$: the complete graph $K_{|T|}$ where each vertex corresponds to a terminal in $T$, and each edge has weight equal to the length of the shortest path between corresponding terminals. An MST over $\tilde{G}$ induces $|T|-1$ shortest paths in $G$; combining all induced paths and removing cycles yields a feasible Steiner tree whose cost is at most $2\left(1 - \frac{1}{|T|}\right)$ times the optimum.

For computing an MST over $\tilde{G}$, one can use any known MST algorithm (e.g., Kruskal's, Prim's, or Bor\r{u}vka's algorithm). However, one can directly construct a Steiner tree from scratch based on these MST algorithms without the need to construct $\tilde{G}$; Poggi de Arag\~{a}o and Werneck provide details for such implementations \cite{de2002implementation} (see also \cite{takahashi1980approximate,wu1986faster}).

Specifically, the Prim-based approximation algorithm for the ST problem due to Takahashi and Matsuyama \cite{takahashi1980approximate} grows a tree rooted at a fixed terminal. In each iteration, the closest terminal not yet connected to the tree is connected through its shortest path. The process continues for $|T|-1$ iterations until all terminals are spanned. The resulting Steiner tree achieves the $2\left(1 - \frac{1}{|T|}\right)$ approximation guarantee \cite{takahashi1980approximate}. The Kruskal-based algorithm for the ST problem due to Wang \cite{wang1985multiple} maintains a forest initially containing $|T|$ singleton trees. In each iteration, the closest pair of trees is connected via a shortest path between them. The process continues for $|T|-1$ iterations until the resulting forest is a tree.  Widmayer showed that this algorithm achieves the $2\left(1 - \frac{1}{|T|}\right)$ bound \cite{widmayer1986approximation}.

\subsection{Review of the QoSMT algorithm of Charikar et al.}\label{section:charikar}

Charikar et al.~\cite{Charikar2004ToN} give a $\min\{2(\ln|T|+1), \ell \rho\}$-approximation for QoSMT which we denote by $\mathrm{C_2}$, consisting of taking the better of the solutions returned by two sub-algorithms (denoted $\mathrm{C_{2a}}$ and $\mathrm{C_{2b}}$). For this section, we focus primarily on the $2(\ln|T|+1)$-approximation, Algorithm~\QoSNU.  The $\ell \rho$-approximation, Algorithm~$\mathrm{C_{2b}}$, simply computes a $\rho$-approximate Steiner tree over the terminals of each priority separately, then merges the $\ell$ computed trees and prunes cycles to output a tree; this leads to a better approximation ratio if $\ell \ll |T|$.

The first sub-algorithm (\QoSNU) sorts the terminals $T$ by decreasing priority $P(\cdot)$, starting with a root node $r$ (here, we may treat the root as any terminal with priority $\ell$). Then, for $i=1,\ldots, |T|$, the $i^{\text{th}}$ terminal $t_i$ is connected to the existing tree spanning the previous $i-1$ terminals using the minimum cost path with edges of rate at least $P(t_i)$, where the cost of this path is defined as the \emph{connection cost} of $t_i$. 

The authors show that for $1 \le m \le |T|$, the $m^{\text{th}}$ most expensive connection cost is at most $\frac{2\OPT}{m}$, which implies that the total cost is at most $2\OPT\left(1+\frac{1}{2}+\frac{1}{3} + \ldots + \frac{1}{|T|}\right)\le 2(\ln |T|+1) \OPT$. While not explicitly mentioned in~\cite{Charikar2004ToN}, this approximation ratio is roughly tight (see Figure~\ref{fig:kruskal-example-tightness}). Algorithm~\QoSNU can be implemented by running Dijkstra's algorithm from $t_i$ 
until a vertex already in 
the tree is encountered.  The running time of \QoSNU is roughly $|T|$ times the running time of Dijkstra's algorithm, or $O(nm+n^2 \log n)$~\cite{Charikar2004ToN}.

\section{Kruskal-based MLST algorithms}
\label{section:kruskal}
We propose Algorithm~\Kruskalu for the MLST problem. The main  distinction compared to Algorithm \QoSNU is that the subsequent algorithm connects the ``closest'' pairs of terminals first, rather than connecting terminals in order of priority. 
Algorithm~\Kruskalu proceeds as follows: initializing $S=T$, while $|S| \ge 1$, find terminals $u, v \in S$ with $P(u) \ge P(v)$ which minimize the cost of connecting them. If $\mathcal{P}$ is the $u$--$v$ path chosen, then the rate of each edge in $\mathcal{P}$ is upgraded to $P(v)$ (if its rate is less). Remove $v$ from $S$. We will say that $v$ is \emph{connected} at the current iteration. When $|S|=1$, if there are no cycles, then the resulting tree is a feasible MLST rooted at some vertex $r$ with $P(r) = \ell$.  Otherwise, we can prune one edge from each cycle with the lowest rate to produce a tree. We note that \Kruskalu takes $|T|-1$ iterations while \QoSNU takes $|T|$ iterations; this follows as the setting for MLST does not specify a root vertex while QoSMT does. As such, there is a small constant difference in the approximation ratios, which is not significant.

When finding $u, v \in S$ which minimize $\sigma(u,v)$, Algorithm~\Kruskalu takes into account edges which have already been included at lower rates. In other words, line~\ref{line:kruskal}
seeks a pair of vertices $(u,v)$ which minimizes the cost of ``upgrading'' the rates of some edges so that $u$ and $v$ are connected via a path of rate $\min\{P(u), P(v)\}$. We denote this cost by $\sigma'(u,v)$, and observe that $\sigma'(u,v) \le \sigma(u,v)$.
\begin{algorithm}[h!]
\renewcommand{\thealgorithm}{}
\caption{\textsc{KruskalMLST}$(\text{graph }G, \text{ priorities }P, \text{ costs }c)$}
\begin{algorithmic}[1]
\State Initialize $y(e) = 0$ for $e \in E$
\State $c'_i(e)=c_i(e)$ for $i\in[\ell],e\in E$
\State $S = T$
\While{$|S|>1$}\label{line:kruskal-while}
\State Compute $\sigma'(\cdot,\cdot)$ for all $(\cdot,\cdot) \in S \times S$
\State Find $u,v \in S$ with $P(u) \ge P(v)$ which minimizes $\sigma'(u,v)$\label{line:kruskal}
\State $\mathcal{P} = $ path chosen of cost $\sigma'(u,v)$
\State $y(e) = \max\{y(e), P(v)\}$ for $e \in \mathcal{P}$
\State $c'_i(e)=\max\{0,c_i(e)-c_{y(e)}(e)\}$ for $e \in \mathcal{P}$ and $i\in \{1,\dots,\ell\}$
\State $S = S\setminus \{v\}$
\EndWhile
\State \Return{$y$}
\end{algorithmic}
\end{algorithm}

\begin{theorem}\label{thm:2lnt}
Algorithm~\Kruskalu is a $2\ln|T|$-approximation to the MLST problem. 
\end{theorem}
\begin{proof}
Define the \emph{connection cost} of $v$ to be $\sigma'(u,v)$ (line 6), and note that the cost of the returned solution is the sum of the connection costs over all terminals $T\setminus \{r\}$.   Now let $t_1$, $t_2$, \ldots, $t_{|T|-1}$ be the terminals in sorted order by which they were connected, and let $\OPT$ denote the cost of a minimum cost MLST for the instance. We have the following lemma.

\begin{lemma}\label{lemma:2opt}
For $2 \le m \le |T|$, consider the iteration of Algorithm~\Kruskalu when $|S|=m$. Let $t_i$ be the terminal connected during this iteration (where $i = |T|+1-m$). Then the connection cost of $t_i$ is at most $\frac{2\OPT}{m}$.
\end{lemma}
\begin{proof}
Note that immediately before $t_i$ is connected, we have $S = \{t_i, t_{i+1}, \ldots, t_{|T|-1}, r\}$ of size $m$. Consider the optimum solution $\mathcal{T}^*$ for the instance, and let $\mathcal{T}'$ be the minimal subtree of $\mathcal{T}^*$ containing all terminals in $S$. The total cost of the edges in $\mathcal{T}'$ is at most $\OPT$. Perform a depth-first traversal starting from any terminal in $\mathcal{T}'$ and returning to that terminal. Since every edge in $\mathcal{T}'$ is traversed twice, the cost of the traversal is at most $2\OPT$.

Consider pairs of consecutive terminals $t_j$, $t_k$ visited for the first time along the traversal. The path connecting $t_j$ and $t_k$ in $\mathcal{T}^\prime$ necessarily uses edges of rate at least $\min\{P(t_j), P(t_k)\}$. Then, the cost of the edges along this path is at least $\sigma(t_j, t_k)$. There are $m$ pairs of consecutive terminals along the traversal (including the pair containing the first and last terminals visited), and the sum of the costs of these $m$ paths is at most $2\OPT$. Hence, some pair $t_j$, $t_k$ of terminals is connected by a path of cost $\le \frac{2\OPT}{m}$ in the optimum solution, implying that for this pair $t_j, t_k$, we have $\sigma'(t_j, t_k) \le \sigma(t_j, t_k) \le \frac{2\OPT}{m}$. Since~\Kruskalu selects the pair which minimizes $\sigma'(\cdot, \cdot)$, the connection cost of $t_i$ is at most $\frac{2\OPT}{m}$.
\end{proof}

Lemma~\ref{lemma:2opt} immediately implies Theorem~\ref{thm:2lnt}. Indeed, summing from $m=2$ to $m=|T|$, the total cost is at most $2\OPT\left(\frac{1}{2} + \frac{1}{3} + \ldots + \frac{1}{|T|}\right) = 2\OPT(H_{|T|}-1) \le 2\ln |T| \OPT$.
\end{proof}

An interesting note is that Algorithm \Kruskalu reduces to the Kruskal-based algorithm~\cite{wang1985multiple} for computing a Steiner tree, when there are no priorities on the terminals (i.e., the single level case when $\ell=1$).  As mentioned earlier, this is a $2(1-\frac{1}{|T|})$-approximation, whereas algorithm \QoSNU is still a  $2\ln|T|$ one, and this is an advantage of the proposed algorithm.

A simple variant of our algorithm, \Kruskalnu, yields the same theoretical approximation ratios and is easier to implement. The difference is that \Kruskalnu does not update the costs $\sigma$ at each iteration of the while loop.

\begin{algorithm}[h!]
\renewcommand{\thealgorithm}{}
\caption{$\Kruskalnu(\text{graph }G, \text{ priorities }P, \text{ costs }c)$}
\begin{algorithmic}[1]
\State Initialize $y(e) = 0$ for $e \in E$
\State $S = T$
\While{$|S|>1$}\label{line:kruskalnu-while}
\State Find $u,v \in S$ with $P(u) \ge P(v)$ which minimizes $\sigma(u,v)$\label{line:kruskalnu}
\State $\mathcal{P} = $ path chosen of cost $\sigma(u,v)$
\State $y(e) = \max(y(e), P(v))$ for $e \in \mathcal{P}$
\State $S = S\setminus \{v\}$
\EndWhile
\State \Return{$y$}
\end{algorithmic}
\end{algorithm}

\begin{theorem}
Algorithm \Kruskalnu is a $2\ln|T|$-approximation to the MLST problem.
\end{theorem}

The proof follows the same argument as that for Theorem \ref{thm:2lnt}; indeed the use of $\sigma'$ implies that \Kruskalu should perform better than \Kruskalnu, but is more costly to run.

\subsection{Tightness} 
The approximation ratio for Algorithms~\QoSNU~\cite{Charikar2004ToN} and~\Kruskalnu is tight up to a constant, even if $\ell = 1$ or if $|E|=O(|V|)$. As a tightness example, we use a graph construction $(G_i)_{i \ge 0}$ given by Imase and Waxman~\cite{imase91dst} for the inapproximability of the dynamic Steiner tree problem. Let $G_0$ contain two vertices $v_0$, $v_1$ with an edge of cost 1 connecting them. We say that $v_0$ and $v_1$ are depth zero vertices. For $i \ge 1$, graph $G_i$ is obtained by replacing each edge $uv$ in $G_{i-1}$ with two depth $i$ vertices $w_1$, $w_2$, and adding edges $uw_1$, $w_1 v$, $uw_2$, and $w_2 v$.

Let $G = G_k$ for sufficiently large $k$, let $\ell=1$ (i.e., the Steiner tree problem), and let each edge of $G_i$ have a cost of $\frac{1}{2^i}$, so that the cost of any shortest $v_0$-$v_1$ path is 1. Let the terminals $T$ be the vertices of some $v_0$-$v_1$ path (Figure~\ref{fig:kruskal-example-tightness}, left), so that $\OPT = 1$. Note that any $u$-$v$ path contains $2^k$ edges, so $|T|=2^k + 1$. Algorithm~\QoSNU first sorts the terminals by priority; since all terminals in $G_k$ have the same priority, we consider a worst possible ordering where $T$ is ordered in increasing depth, with $v_0$ the root. In this case, it is possible that Algorithm~\QoSNU connects $v_1$ to $v_0$ via a shortest path which does not include other terminals, then connects subsequent terminals via shortest paths which include no other terminal, as shown in Figure~\ref{fig:kruskal-example-tightness}. Conversely in the worst case, Algorithm~\Kruskalnu may connect depth $k$, $k-1$, $k-2$, \ldots terminals in order while avoiding previously-used paths, as Algorithm~\Kruskalnu does not consider existing edges. In both cases, the cost of the returned solution is
\[\textrm{Cost} = \frac{1}{2}k + 1 = \frac{1}{2}\log_2 (|T|-1) + 1 \ge \frac{1}{2}\left(\log_2 |T| + 1\right) \OPT \approx \left(0.72 \ln |T| + \frac{1}{2}\right)\OPT.\]

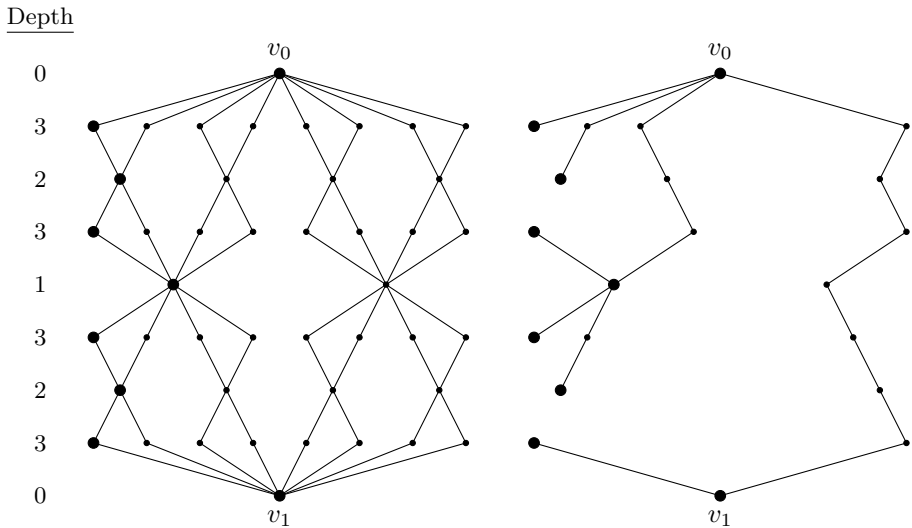
\begin{figure}[h!]
    \centering
    \begin{subfigure}[b]{0.49\textwidth}
    \begin{tikzpicture}[scale=0.7]
    
	\foreach \y in {4,-4}{
		\draw [fill=black] (0, \y) circle [radius=0.1];
	}
	\foreach \x in {-2,2}{
		\ifthenelse{\x=-2}
		{\def\r{0.1}}
		{\def\r{0.05}}
		\draw [fill=black] (\x, 0) circle [radius=\r];
	}
	\foreach \x in {-3,-1,1,3}{
		\ifthenelse{\x=-3}
		{\def\r{0.1}}
		{\def\r{0.05}}
		\foreach \y in {2,-2}{
			\draw [fill=black] (\x, \y) circle [radius=\r];
		}
		\draw (\x, 2) -- (\x+0.5, 3);
		\draw (\x, 2) -- (\x-0.5, 3);
		\draw (\x, -2) -- (\x+0.5, -3);
		\draw (\x, -2) -- (\x-0.5, -3);
		\draw (\x, 2) -- (\x+0.5, 1);
		\draw (\x, 2) -- (\x-0.5, 1);
		\draw (\x, -2) -- (\x+0.5, -1);
		\draw (\x, -2) -- (\x-0.5, -1);

	}
	\foreach \x in {-4,...,3}{
		\ifthenelse{\x=-4}
		{\def\r{0.1}}
		{\def\r{0.05}}
		\foreach \y in {3,1,-1,-3}{
    		\draw [fill=black] (\x+0.5, \y) circle [radius=\r];
    	}
    \ifthenelse{\x < 0}
	{\draw (\x+0.5,1) -- (-2,0);
	\draw (\x+0.5,-1) -- (-2,0);}
	{\draw (\x+0.5,1) -- (2,0);
	\draw (\x+0.5,-1) -- (2,0);}
	}

	\foreach \x in {-4,...,3}{
		\draw (0,4) -- (\x+0.5,3);
		\draw (0,-4) -- (\x+0.5,-3);
	}
	\node[above] at (0,4.1) {$v_0$};
	\node[below] at (0,-4.1) {$v_1$};
	\node at (-4.5,5) {\underline{\small Depth}};
	\node at (-4.5,4) {\small 0};
	\node at (-4.5,3) {\small 3};
	\node at (-4.5,2) {\small 2};
	\node at (-4.5,1) {\small 3};
	\node at (-4.5,0) {\small 1};
	\node at (-4.5,-1) {\small 3};
	\node at (-4.5,-2) {\small 2};
	\node at (-4.5,-3) {\small 3};
	\node at (-4.5,-4) {\small 0};
    \end{tikzpicture}
\end{subfigure}
\begin{subfigure}[b]{0.49\textwidth}
    \begin{tikzpicture}[scale=0.7]
	\foreach \x\y in {(0,4), (-3.5,3), (-3,2), (-3.5,1), (-2,0),
	(-3.5,-1), (-3,-2), (-3.5,-3), (0,-4)}{
		\draw [fill=black] \x\y circle [radius=0.1];
	}
	\node[above] at (0,4.1) {$v_0$};
	\node[below] at (0,-4.1) {$v_1$};
	
	\foreach \x\y in {(-3.5,3),(-2.5,3),(-1.5,3),(3.5,3),(-1,2),(3,2),(-0.5,1),(3.5,1),
	(2,0), (-2.5,-1), (2.5,-1), (3,-2), (3.5,-3)}{
		\draw [fill=black] \x\y circle [radius=0.05];
	}
	\draw (-3.5,3) -- (0,4);
	\draw (-3.5,1) -- (-2,0);
	\draw (-3.5,-1) -- (-2,0);
	\draw (-3.5,-3) -- (0,-4);
	
	\draw (-3,2) -- (-2.5,3) -- (0,4);
	\draw (-3,-2) -- (-2.5,-1) -- (-2,0);
	
	\draw (-2,0) -- (-0.5,1) -- (-1,2) -- (-1.5,3) -- (0,4);
	
	\draw (0,-4) -- (3.5,-3) -- (3,-2) -- (2.5,-1) -- (2,0) -- (3.5,1) -- (3,2) -- (3.5,3) -- (0,4);
    \end{tikzpicture}
\end{subfigure}
    \caption{\emph{Left:} Example instance where $G = G_3$ using the construction by Imase and Waxman~\cite{imase91dst}, $\ell=1$, with terminals bolded. All edges have cost $\frac{1}{8}$ so that $\OPT=1$. \emph{Right:} Example solution $\mathcal{T}$ which could be returned by Algorithms~\QoSNU and~\Kruskalnu, with cost $\frac{20}{8}$. Note that in hindsight, $G$ may be sparsified so that $|E| = O(|V|)$, by letting $E = E(\mathcal{T}) \cup E(\mathcal{T}^*)$, then contracting each simple path between two terminals to a single edge with cost equal to the length of the path.}
    \label{fig:kruskal-example-tightness}
\end{figure}

\subsection{Running Time} The running time of Algorithm~\Kruskalnu is similar to that of Algorithm~\QoSNU, namely $|T|$ times the running time of Dijkstra's algorithm. This can be implemented as follows: before line~\ref{line:kruskal-while}, for each terminal $t \in T$, run Dijkstra's algorithm from $t$ using edge weights $c_{P(t)}(\cdot)$, and only keep track of distances from $t$ to terminals with priority $\ge P(t)$. Thus, each terminal $t\in T$ keeps a dictionary of distances from $t$ to a subset of $T$. Then at each iteration (line~\ref{line:kruskal}), find the minimum distance among at most $|T|$ distances. The running time of \Kruskalu is $|T|^2$ times that of Dijkstra's algorithm due to the update step. 




\section{Prim-based MLST algorithm}


A natural approach based on Prim's algorithm is as follows. Choose a root terminal $r$ with $P(r) = \ell$ and remove $r$ from $T$. Then, find a terminal $v \in T$ whose connection cost is minimum, where the connection cost is defined to be the cost of installing or upgrading edges from $r$ to $v$ using rate $P(v)$ (namely, using edge costs $c_{P(v)}(\cdot)$). Remove $v$ from $T$, and decrement costs. Repeat this process of connecting the existing MLST to the closest terminal until $T$ is empty. Interestingly, unlike Algorithm~\Kruskalnu, this approach can return a solution $|T|$ times the optimum, which is rather poor.  We remark that Algorithm~\QoSNU~\cite{Charikar2004ToN} is similar to the Prim-based algorithm, where terminals are connected in order of priority rather than connecting the closest terminals first.

As an example, suppose $G$ is a cycle containing $|V|=\ell+1$ vertices $v_1$, $v_2$, $v_3$, \ldots, $v_{\ell}$, $r$ in that order (Figure~\ref{fig:prim}, left). Let $P(v_i) = i$, and let $P(r) = \ell$. Let $c_i(rv_{\ell}) = 1$ (edge $rv_{\ell}$ has cost 1 regardless of rate), and let $c_i(rv_1) = i(1-\varepsilon)$. Let all other edges have cost zero (or perhaps a small $\varepsilon' \ll \varepsilon$), regardless of rate. Then the Prim-based algorithm greedily connects $v_1$, $v_2$, \ldots, $v_\ell$ in that order, incurring a cost of $1-\varepsilon$ at each iteration. Hence the cost returned is $\ell(1-\varepsilon) \approx |T|$, while $\OPT=1$.


\begin{figure}[h]
    \centering
    \begin{subfigure}[b]{0.3\textwidth}
\centering
\begin{tikzpicture}[scale=0.8]
\draw (-1.4,0) -- (0,0) -- (1,1) -- (1,2.4) -- (0,3.4);
\draw[dashed] (0,3.4) .. controls (-1.4,4) and (-4,1) .. (-1.4,0);
\node[circle,fill=black!40] at (0,0) {\footnotesize $\ell$};
\node[circle,fill=black!10] at (1,1) {\footnotesize 1};
\node[circle,fill=black!10] at (1,2.4) {\footnotesize 2};
\node[circle,fill=black!10] at (0,3.4) {\footnotesize 3};
\node[circle,fill=black!10] at (-1.4,0) {\footnotesize $\ell$};
\node[below] at (-0.7,0) {1};
\node[below,right] at (0.5,0.4) {\footnotesize $i(1-\varepsilon)$};
\end{tikzpicture}
\end{subfigure}
\hfill
\begin{subfigure}[b]{0.3\textwidth}
\centering
\begin{tikzpicture}[scale=0.8]
\draw (-1.4,0) -- (0,0);
\draw[line width=0.5mm] (0,0) -- (1,1) -- (1,2.4) -- (0,3.4);
\draw[dashed, line width=0.5mm] (0,3.4) .. controls (-1.4,4) and (-4,1) .. (-1.4,0);
\node[circle,fill=black!40] at (0,0) {\footnotesize $\ell$};
\node[circle,fill=black!10] at (1,1) {\footnotesize 1};
\node[circle,fill=black!10] at (1,2.4) {\footnotesize 2};
\node[circle,fill=black!10] at (0,3.4) {\footnotesize 3};
\node[circle,fill=black!10] at (-1.4,0) {\footnotesize $\ell$};
\node[below] at (-0.7,0) {1};
\node[below,right] at (0.5,0.4) {\footnotesize $i(1-\varepsilon)$};
\end{tikzpicture}
\end{subfigure}
\hfill
\begin{subfigure}[b]{0.3\textwidth}
\begin{tikzpicture}[scale=0.8]
\draw (0,0) -- (1,1);
\draw[line width=0.5mm] (-1.4,0) -- (0,0);
\draw[line width=0.5mm] (1,1) -- (1,2.4) -- (0,3.4);
\draw[dashed, line width=0.5mm] (0,3.4) .. controls (-1.4,4) and (-4,1) .. (-1.4,0);
\node[circle,fill=black!40] at (0,0) {\footnotesize $\ell$};
\node[circle,fill=black!10] at (1,1) {\footnotesize 1};
\node[circle,fill=black!10] at (1,2.4) {\footnotesize 2};
\node[circle,fill=black!10] at (0,3.4) {\footnotesize 3};
\node[circle,fill=black!10] at (-1.4,0) {\footnotesize $\ell$};
\node[below] at (-0.7,0) {1};
\node[below,right] at (0.5,0.4) {\footnotesize $i(1-\varepsilon)$};
\end{tikzpicture}
\end{subfigure}
    \caption{\emph{Left:} Simple example demonstrating that a Prim-based algorithm can perform poorly. The priorities $P(\cdot)$ and edge costs $c_i(\cdot)$ are shown, and the root $r$ is bolded. \emph{Center:} Solution found by the Prim-based algorithm with cost $\ell(1-\varepsilon)$. \emph{Right:} Optimum solution with cost $\OPT=1$.}
    \label{fig:prim}
\end{figure}
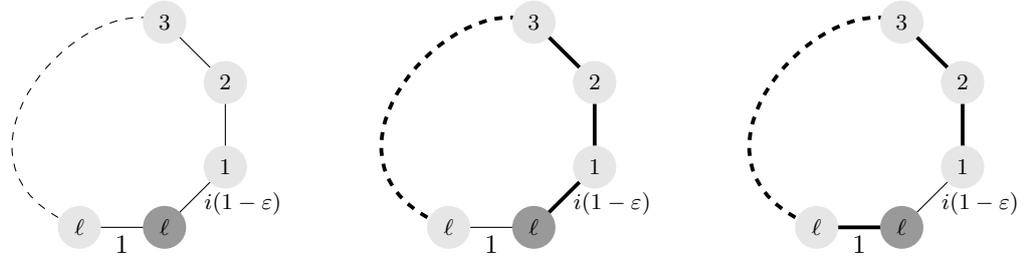

\section{Integer linear programming (ILP) formulation}
\label{section:exact_algorithm}

In~\cite{MLST2018}, ILP formulations were given for the MLST problem with proportional costs. 
We extend these and give an ILP formulation for non-proportional costs. First, direct the graph $G$ by replacing each edge $e=uv$ with two directed edges $(u,v)$ and $(v,u)$. Let $x_{uv}^i = 1$ if $(u,v)$ appears in the solution with rate greater than or equal to $i$, and 0 otherwise. Let $c'_i(u,v)$ denote the incremental cost of edge $(u,v)$ with rate $i$, defined as $c_i(e)-c_{i-1}(e)$ where $e = uv$ and $c_0(e) = 0$. Fix a root $r \in T$ with $P(r) = \ell$. For $i = 1, \ldots, \ell$, let $T_i = \{t \in T: P(t) \ge i\}$ denote the set of terminals requiring priority at least $i$.  For every edge $e=(u,v)$ we define two flow variables $f_{uv}^i$ and $f_{vu}^i$.

\begin{align}
\text{Minimize }& \hspace{3.8ex} \sum \limits_{i=1}^{\ell} \sum \limits_{(u,v) \in E}
c'_i(u,v) x_{uv}^i \text{ subject to} \\
\sum \limits_{(v,w) \in E} f_{vw}^i - \sum \limits_{(u,v) \in E} f_{uv}^i & = 
\begin{cases}
  |T_i|-1 & \quad \text{if } v=r\\
  -1      & \quad \text{if } v \in T_i \setminus \{r\} \\
  0       & \quad \text{else}\\
\end{cases} && \forall \, v \in V; 1 \le i \le \ell \\
x_{uv}^{i} & \le x_{uv}^{i-1} && \forall \, (u,v) \in E; 2 \le i \le \ell \\
0 \leq f_{uv}^i & \leq (|T_i|-1) \cdot x_{uv}^i 
&& \forall \, (u,v) \in E; 1 \le i \le \ell \\
x_{uv}^i &\in \{0, 1\} && \forall \, (u,v) \in E; 1 \le i \le \ell
\end{align}

In the optimal solution, the edges of rate greater than or equal to $i$ form a Steiner tree over $T_i$, so the flow constraint ensures that this property holds. The second constraint ensures that if an edge is selected at rate $i$ or greater, then it must be selected at lower rates. The third constraint ensures that the indicator variable is set equal to one if and only if the corresponding edge is in a tree. The last constraint ensures that the $x_{uv}^i$ variables are 0--1.

\begin{theorem}\label{thm:ilp}
The optimal solution for the ILP induces an MLST with cost $\OPT$.
\end{theorem}
The proof is deferred to Appendix~\ref{apdx:ilp}. Additionally, it can be seen from the formulation that the number of variables is $O(\ell |E|)$ and the number of constraints is $O(\ell (|E|+|V|))$.

\section{Experiments}\label{section:experiment}

We run two primary kinds of experiments: first, we compare the various MLST approximation algorithms discussed here on random graphs from different generators; second, to provide comparison with the Steiner tree literature, we perform experiments on instances generated using the SteinLib library~\cite{KMV00}.  In both cases, we consider natural questions about how the number of priorities, number of vertices, and decay rate of terminals with respect to priorities affect the running times and (experimental) approximation ratios (cost of returned solution divided by $\OPT$) of the algorithms explored here.  We also record how often the algorithms proposed here provide better approximation ratios than pre-existing algorithms.  Moreover, we illustrate a class of graphs for which Algorithm \Kruskalu always performs better than Algorithm \QoSNU.

\subsection{Experiment Parameters}

We run experiments first to test runtime vs. parameters discussed above, and then to test the experimental approximation ratio vs. the parameters.  Each set of experiments has several parameters: the graph generator (random generators or SteinLib instances), the maximum number of priorities $\ell$, $|V|$, how the size of the terminal sets $T_i$ (terminals requiring priority at least $i$) decrease as $i$ decreases, and proportional vs. non-proportional edge costs.

In what follows, we use the Erd\H{o}s--R\'{e}nyi (ER)~\cite{erdos1959random}, Watts--Strogatz (WS)~\cite{watts1998collective}, and Barab\'{a}si--Albert (BA)~\cite{barabasi1999emergence} models or SteinLib instances~\cite{KMV00} to generate the input graph (more on how SteinLib instances are given priorities later).  We consider number of priorities $\ell\in\{2,\dots,7\}$, and adopt two methods for selecting terminal sets (equivalently priorities): \emph{linear} and \emph{exponential}.  A terminal set $T_\ell$ with lowest priority of size $n(1-\frac{1}{\ell+1})$ in the linear case and $\frac{n}{2}$ in the exponential case is chosen uniformly at random.  For each subsequent priority, $\frac{1}{\ell+1}$ terminals are deleted at random in the linear case, whereas half the remaining terminals are deleted in the exponential case.  Priorities and terminal sets are related via $T_i = \{t \in T: P(t) \ge i\}$. For the proportional edge weight case, we choose $c_1(e)$ uniformly at random from $\{1,\dots,10\}$ for each edge independently and set $c_i(e)=ic_1(e)$ for $i=1,\dots,\ell$.  For the non-proportional setting, we select the incremental edge costs $c_1(e)$, $c_2(e) - c_1(e)$, $c_3(e) - c_2(e)$, \ldots, $c_\ell(e) - c_{\ell-1}(e)$ uniformly at random from $\{1,2,3,\ldots,10\}$ for each edge independently.

In the case that the input graph comes from SteinLib, it has a prescribed terminal set (since SteinLib graphs are instances of ST problem for a single priority).  For these inputs, priorities are generated in two ways: filtered terminals and augmented terminals. To generate filtered terminals we divide the set of original terminals from the SteinLib into $\ell$ sets (with $\ell\in\{2,\dots,6\}$). We assign the first set as the topmost priority terminals. We assign the second set to the next priority and so on. For the augmented case, we start with the initial terminals from the SteinLib instance and add additional terminals uniformly at random from the remaining vertices.  We assign $5$ vertices as top priority terminals, double the number of terminals in the next priority, and so on until the maximum number of terminals is reached (we assign $\ell\in\{2,3,4\}$ priorities). Augmentation makes sense given that some of the original SteinLib instances have very few terminals. We have generated our datasets from two subsets of SteinLib: \href{http://steinlib.zib.de/showset.php?I080}{I080} and \href{http://steinlib.zib.de/showset.php?I160}{I160}; we generate both types of terminals (filtered and augmented) for each of these datasets.

An experimental instance of the MLST problem here is thus characterized by five parameters: graph generator, number of vertices $|V|$, number of priorities $\ell$, terminal selection method $\textsc{TSM}\in\{\textsc{Linear,Exponential}\}$, and proportionality of the edge weights $\textsc{TE}\in\{\textsc{Prop,Non-prop}\}$.  As there is randomness involved, we generated five instances for every choice of parameters (e.g., ER, $|V| = 70$, $\ell=4$, \textsc{Linear}, \textsc{Non-prop}). 

For the following experiments, we implement the \Kruskalu and \QoSU algorithms in the proportional case, and the \Kruskalu and \QoSNU algorithms in the non-proportional case. We note here that Algorithm \Kruskalnu achieves much poorer results with respect to OPT than \Kruskalu in practice despite having similar theoretical guarantees.  To compute the approximation ratios, we use the ILP described in Section \ref{section:exact_algorithm} using CPLEX 12.6.2 as an ILP solver.

\subsection{Results}

As one would expect, runtime for both the ILP and all approximation algorithms increased as $|V|$ or $\ell$ increased.  Runtime was typically higher for linear terminal selection than for exponential.  See Figures \ref{BoxPlots_ER_uniform_time_heu}--\ref{BoxPlots_BA_uniform_time_heu} in the Appendix for detailed plots. We do note that the running times of the approximation algorithms are significantly faster than the running time of the ILP; the latter takes a couple of minutes for whereas the approximation algorithms take only a couple of seconds for the same instances generated in our experiments.

There was no discernible trend in plots of Ratio (defined as cost/\OPT) vs. $|V|$, $\ell$, or the terminal selection method (linear or exponential).  In all cases, for all graph generators, both the \Kruskalu and \QoSU (or \QoSNU in the non-proportional case) exhibited similar statistical behavior independent of the given parameter (see Figures \ref{BoxPlots_WS_uniform_cmp}--\ref{BoxPlots_BA_nonuniform} in the Appendix for detailed plots).  For a brief illustration, we show the behavior for Erd\H{o}s--R\'{e}nyi graphs with $p=(1+\varepsilon)\frac{\ln n}{n}$ in Figure \ref{BoxPlots_ER_uniform_cmp}, and include the performance of the Composite Algorithm of \cite{MLST2018} (CMP) as it gives the best \textit{a priori} approximation ratio guarantee. 

\begin{figure}[ht]
\begin{minipage}{\textwidth}
    \centering
    \begin{subfigure}[b]{0.30\textwidth}
        \includegraphics[width=\textwidth]{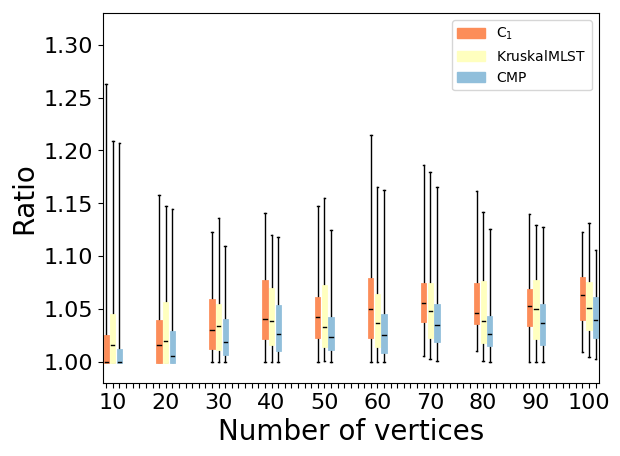}
    \end{subfigure}
    ~ 
    \begin{subfigure}[b]{0.30\textwidth}
        \includegraphics[width=\textwidth]{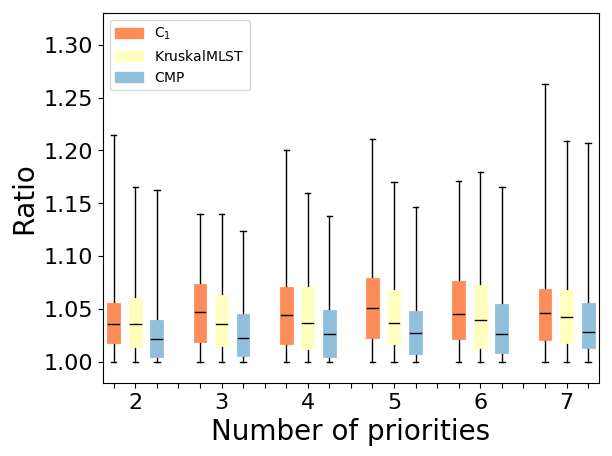}
    \end{subfigure}
    ~
    \begin{subfigure}[b]{0.30\textwidth}
        \includegraphics[width=\textwidth]{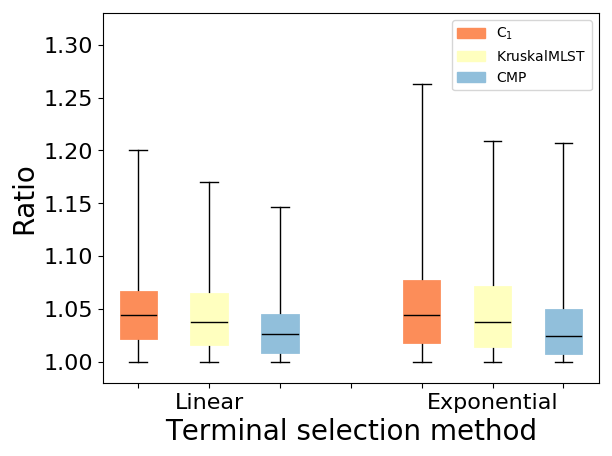}
    \end{subfigure}
    \caption{Performance of \QoSU~\cite{Charikar2004ToN}, \Kruskalu, and \CMP~\cite{MLST2018} on Erd\H{o}s--R{\'e}nyi graphs w.r.t.\
      $|V|$, $\ell$, and terminal selection method with proportional edge weights.}
    \label{BoxPlots_ER_uniform_cmp}
\end{minipage}
\end{figure}

From this figure, we see that on average \Kruskalu outperforms \QoSNU.  However, it is instructive to compare the instance-wise performance of the different algorithms.  Tables \ref{table:proportionalKeaton} and \ref{table:non_proportionalKeaton} show comparisons of the statistical performance of the the two approximation algorithms for various graph generators in the proportional and non-proportional case, respectively. For each graph generator, there are a total of 1140 instances consisting of 5 graphs for each set of parameters ($|V|$, $\ell$, etc.).

\begin{table}[h]
\resizebox{\columnwidth}{!}{
\begin{tabular}{|c||c|c||c|c||c|c||c|c|}
\hline
Graph Generator & \multicolumn{2}{c||}{ER} & \multicolumn{2}{c||}{WS} & \multicolumn{2}{c||}{BA} & \multicolumn{2}{c|}{SteinLib}\\ \hline
Algorithm & \QoSU & K & \QoSU & K & \QoSU & K & \QoSU & K \\ \hline
Equal to \OPT & 73 & \textbf{133} & 391 & \textbf{679} & 94 & \textbf{202} & 4 & \textbf{8} \\
\hline
Mean & 1.048 & \textbf{1.044} & 1.016 & \textbf{1.012} & 1.028 & \textbf{1.021} & 1.2355 & \textbf{1.1918} \\
\hline
Median & 1.044 & \textbf{1.037} & 1.006 & \textbf{1.0} & 1.019 & \textbf{1.016} & 1.2072 & \textbf{1.1707} \\
\hline
Min & \textbf{1.0} & \textbf{1.0} & \textbf{1.0} & \textbf{1.0} & \textbf{1.0} & \textbf{1.0} & \textbf{1.0} & \textbf{1.0} \\
\hline
Max & 1.263 & \textbf{1.202} & 1.31 & \textbf{1.18} & 1.212 & \textbf{1.126} & 1.7488 & \textbf{1.6404} \\
\hline
Best Approx. & 40.53\% & \textbf{54.29\% }& 24.92\% & \textbf{50.78\%} & 30.62\% & \textbf{69.38\%} & 31.50 & \textbf{59.12\%} \\
\hline
\end{tabular}}\caption{Statistics of Algorithms \QoSU~\cite{Charikar2004ToN} and \Kruskalu (abbreviated K) with proportional edge cost. Best Approx.~reports the percentage of instances (out of 1140) that each algorithm achieved strictly better experimental approximation ratio.  Best performance in each category is bolded.}
\label{table:proportionalKeaton}
\end{table}

\begin{table}[h]
\centering
\begin{tabular}{|c||c|c||c|c||c|c|}
\hline
Graph Generator & \multicolumn{2}{c||}{ER} & \multicolumn{2}{c||}{WS} & \multicolumn{2}{c|}{BA}\\ \hline
Algorithm & \QoSNU & K & \QoSNU & K & \QoSNU & K \\ \hline
Equal to \OPT & 16 & \textbf{26} & 16 & \textbf{30} & 10 & \textbf{26} \\
\hline
Mean & 1.123 & \textbf{1.109} & 1.099 & \textbf{1.081} & 1.121 & \textbf{1.097} \\
\hline
Median & 1.109 & \textbf{1.099} & 1.087 & \textbf{1.067} & 1.096 & \textbf{1.08} \\
\hline
Min & \textbf{1.0} & \textbf{1.0} & \textbf{1.0} & \textbf{1.0} & \textbf{1.0} & \textbf{1.0} \\
\hline
Max & 1.667 & \textbf{1.54} & 1.863 & \textbf{1.601} & 1.941 & \textbf{1.667} \\
\hline
Best Approx. & 37.20\% & \textbf{61.22\%} & 34.83\% & \textbf{63.85\%} & 30.62\% & \textbf{68.24\%} \\
\hline
\end{tabular}\caption{Statistics of Algorithms \QoSNU~\cite{Charikar2004ToN} and \Kruskalu (abbreviated K) with non-proportional edge cost. Best Approx.~reports the percentage of instances (out of 1140) that each algorithm achieved strictly better experimental approximation ratio. Best performance in each category is bolded.}
\label{table:non_proportionalKeaton}
\end{table}

We see from these tables that \Kruskalu consistently outperforms the algorithms of \cite{Charikar2004ToN} in each of the statistical categories, and also achieves better instance-wise results a majority of the time, although this behavior depends somewhat on the graph generator.  A full suite of figures is given in the Appendix to further illustrate the performance of each algorithm for the various generators. The trends are essentially the same and are as follows. \Kruskalu outperforms \QoSNU on a majority of instances, but has marginally longer runtime (though the difference is not appreciable); the number of priorities has little effect on runtime or experimental approximation ratio; the number of vertices increases the runtime for some generators, but has little effect on the experimental approximation ratio; experimental approximation ratios are typically better on average for exponentially decreasing terminal sets (which makes sense given that $|T|$ is smaller and the approximation guarantees are $O(\ln|T|)$).  Finally, we note that the Composite algorithm of \cite{MLST2018} can achieve better approximation in the proportional edge cost setting, but is not known to work for the non-proportional setting; additionally Composite suffers from exponential growth in runtime with respect to $\ell$, which is a feature not exhibited by \Kruskalu.

\subsection{Graphs for which \Kruskalu always outperforms \QoSNU}

Here we generate a special class of graphs for which the Kruskal-based algorithm always provides near-optimal solutions, but Algorithm \QoSNU performs poorly. This class of graphs consists of cycles with randomly added edges. Begin with a cycle $v_1, v_2, \cdots, v_n, v_1$ and set the weight of edge $v_1v_n$ be $w - \epsilon$ where length of the path $v_1, v_2, \cdots, v_n$ is $w$. We select $v_1$ and $v_n$ as higher-priority terminals, and the remaining vertices as lower-priority terminals. An algorithm that works in a top-down manner will take the edge $v_1v_n$ for higher priority and pay significantly more than the optimal solution~\cite{MLST2018}. Doing this to every edge $(v_i, v_{i+1})$ results an MLST instance where a top-down approach performs arbitrarily poorly. On these graphs, the algorithm provided in Charikar et al.~\cite{Charikar2004ToN} for proportional instances of MLST 
performs noticeably worse than our Kruskal-based approach (see  Figure~\ref{qos_worst}). 
We generated 500 graphs of this type (augmented with some additional edges at random). The script to generate these graphs are available on Github at \url{ https://github.com/abureyanahmed/Kruskal_based_approximation}. 


\begin{figure}[ht]
\centering
    \includegraphics[width=.40\textwidth]{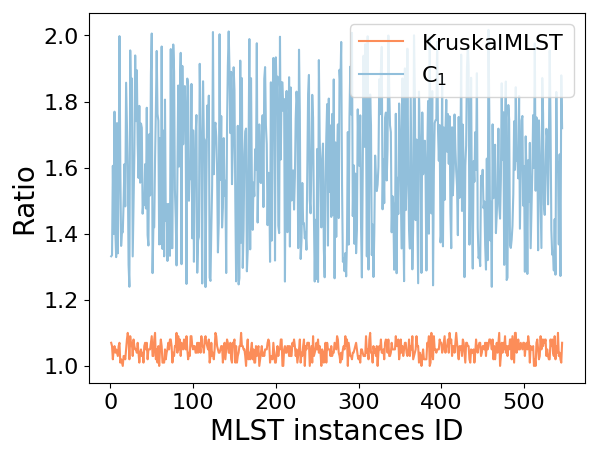}
    \caption{A class of graphs for which the Algorithm \Kruskalu significantly outperforms Algorithm \QoSNU~\cite{Charikar2004ToN}. The $x$--axis is the instance number and carries no meaning of time; the $y$--axis is the approximation ratio.}
    \label{qos_worst}
\end{figure}

\section{Conclusion}
We proposed two algorithms for the MLST problem based on Kruskal's and Prim's algorithms. We showed that the Kruskal-based algorithm is a logarithmic approximation, matching the best approximation guarantee of Charikar et al.~\cite{Charikar2004ToN}, while the Prim-based algorithm can perform arbitrarily poorly. We formulated an ILP for the general MLST problem and provided an experimental comparison between the algorithm provided by Charikar et al.~\cite{Charikar2004ToN}, Ahmed et al.~\cite{MLST2018}, and the Kruskal-based algorithm, \Kruskalu.  We demonstrated that \Kruskalu compares favorably to other algorithms in terms of experimental approximation ratio for both the proportional and non-proportional edge costs while incurring a minor cost in run time.  Finally, we generated a special class of graphs for which~\Kruskalu always performs significantly better than that by Charikar et al.~\cite{Charikar2004ToN}. A natural question is whether the analysis of any of these algorithms~\Kruskalnu, ~\Kruskalu, or~\QoSNU can be tightened, improving the approximability gap between  $O(\log \log n)$ and $O(\log n)$ for the MLST problem with non-proportional edge costs.


\bibliography{mlst}

\newpage
\appendix

\section{Proof of Theorem~\ref{thm:ilp}} \label{apdx:ilp}
\begin{proof}
We first show that the flow variables take only integer values from zero to $|T_i|-1$ although it is not specifically mentioned in the formulation. Note that for every priority the ILP generates a connected component in order to fulfill the conditions of the second equation. The algorithm will compute a tree for every priority, otherwise, there is a cycle at a tree of a particular priority and removing an edge from the cycle minimizes the objective. According to the second equation, the flow variable corresponding to an incoming edge connected to a terminal that is not root is equal to one if the edge is in the tree. Since the difference between the incoming and outgoing flow is $|T_i|-1$ for the root and zero for any intermediate node, every flow variable must be equal to an integer. Also if we do not have integer flows (for example the incoming flow is one and there are two outgoing flows with values 1/2), then because of the conditions in second equation cycles will be generated. Because of this property, the fourth equation ensures that $x_{uv}^i$ is equal to one iff the corresponding flow variable has a value greater than or equal to one. In other words, an indicator variable is equal to one iff the corresponding edge is in the tree. Note that, the formulation has only one assumption on the edge weights: the cost of an edge for a particular rate is greater than or equal to the weight of the edge having lower rates. Hence, the formulation computes the optimal solution for (non-)proportional instances.
\end{proof}

\section{Additional Experimental Results} \label{apdx:experiments}
In this section, we provide some details of the experiments discussed in Section~\ref{section:experiment}.

\subsection{Graph Generator Parameters}
Given a number of vertices, $n$, and probability $p$, the model $\textsc{ER}(n,p)$ assigns an edge to any given pair of vertices with probability $p$.
An instance of $\textsc{ER}(n,p)$ with $p=(1+\varepsilon)\frac{\ln n}{n}$ is connected with high probability for $\varepsilon>0$~\cite{erdos1959random}.  For our experiments we use $n \in \{10, 15, 20, \cdots, 100\}$, and $\varepsilon = 1$.

The Watts--Strogatz model~\cite{watts1998collective} is used to generate graphs that have the small-world property and high clustering coefficient. The model, denoted by $\textsc{WS}(n,K,\beta)$, initially creates a ring lattice of constant degree $K$, and then rewires each edge with probability $0\leq \beta \leq 1$ while avoiding self-loops or duplicate edges. In our experiments, the values of $K$ and $\beta$ are set to $6$ and $0.2$ respectively.

The Barab{\'a}si--Albert model generates networks with
power-law degree distribution, i.e., few vertices become hubs with
extremely large degree~\cite{barabasi1999emergence}. The model is denoted by $\textsc{BA}(m_0,m)$, and uses a preferential
attachment mechanism to generate a growing scale-free network. The
model starts with a graph on $m_0$ vertices. Then, each new vertex
connects to $m\leq m_0$ existing nodes with probability proportional
to its instantaneous degree. This model is a
network growth model. In our experiments, we let the network grow
until the desired network size $n$ is attained. We vary $m_0$ from $10$
to $100$ in our experiments, and set $m=5$.

\subsection{Computing Environment}

For computing the optimum solution, we implemented the ILP described in Section~\ref{section:exact_algorithm} using CPLEX 12.6.2 as an ILP solver.  The model of the HPC system we used for our experiment is Lenovo NeXtScale nx360 M5. It is a distributed system; the models of the processors in this HPC are Xeon Haswell E5-2695 Dual 14-core and Xeon Broadwell E5-2695 Dual 14-core. The speed of a processor is 2.3 GHz. There are 400 nodes each having 28 cores. Each node has 192 GB memory. The operating system is CentOS 6.10.

\subsection{Experimental Setup}

We have considered proportional and non-proportional instances separately. The Kruskal-based algorithm is the same in both settings, but the algorithms of \cite{Charikar2004ToN} admit 2 variants: \QoSU for proportional edge costs which is a $4\rho$--approximation, and \QoSNU for non-proportional edge costs which is a $2(\ln|T|+1)$--approximation.  In figures below, Ratio stands for the approximation ratio given by the cost of the solution returned by the approximation algorithm divided by the optimum cost \OPT returned by the ILP.

All box plots shown below show the minimum, interquartile range (IQR) and maximum, aggregated over all instances using the parameter being compared.

\subsection{Approximation Ratio vs. Parameters -- Proportional edge costs}

First, we take a look at how the approximation ratio of the approximation algorithms is affected by the parameters chosen.  Figures \ref{BoxPlots_ER_uniform_cmp}, \ref{BoxPlots_WS_uniform_cmp}, and \ref{BoxPlots_BA_uniform_cmp} illustrate the change in approximation for different parameters ($|V|$, $\ell$, and the terminal selection method) in the case of proportional edge costs.  For comparison to \cite{MLST2018}, we include the performance of the Composite algorithm (CMP) described therein.

\begin{figure}[h!]
\begin{minipage}{\textwidth}
    \centering
    \begin{subfigure}[b]{0.30\textwidth}
        \includegraphics[width=\textwidth]{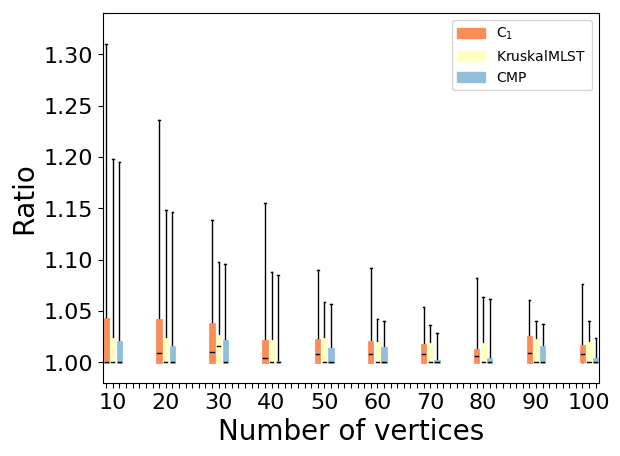}
    \end{subfigure}
    ~ 
    \begin{subfigure}[b]{0.30\textwidth}
        \includegraphics[width=\textwidth]{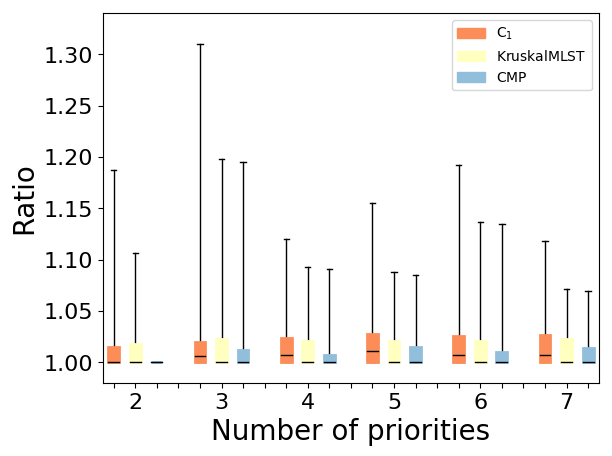}
    \end{subfigure}
    ~
    \begin{subfigure}[b]{0.30\textwidth}
        \includegraphics[width=\textwidth]{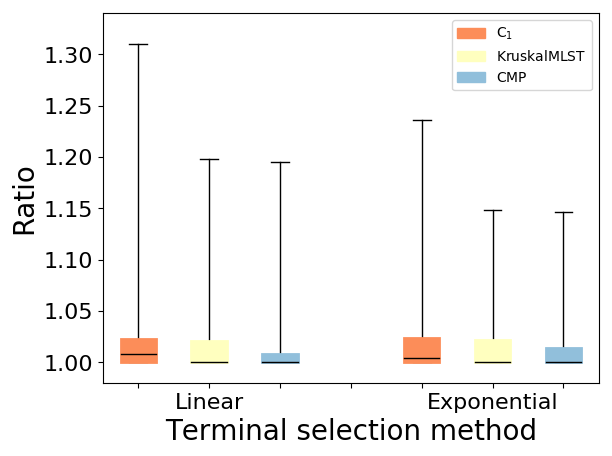}
    \end{subfigure}
    \caption{Performance of \QoSU~\cite{Charikar2004ToN}, \Kruskalu, \CMP~\cite{MLST2018} on Watts--Strogatz graphs w.r.t.\
      $|V|$, $\ell$, and terminal selection method with proportional edge weights.}
    \label{BoxPlots_WS_uniform_cmp}
\end{minipage}
\end{figure}

\begin{figure}[h!]
\begin{minipage}{\textwidth}
    \centering
    \begin{subfigure}[b]{0.30\textwidth}
        \includegraphics[width=\textwidth]{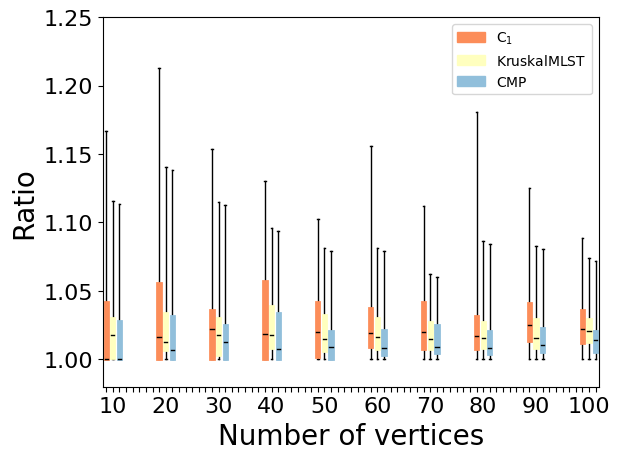}
    \end{subfigure}
    ~ 
    \begin{subfigure}[b]{0.30\textwidth}
        \includegraphics[width=\textwidth]{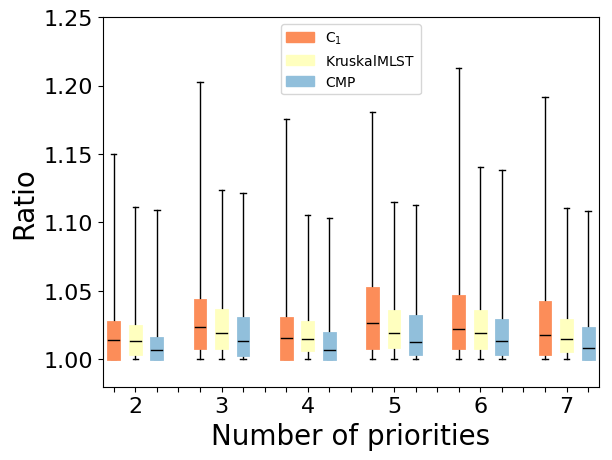}
    \end{subfigure}
    ~
    \begin{subfigure}[b]{0.30\textwidth}
        \includegraphics[width=\textwidth]{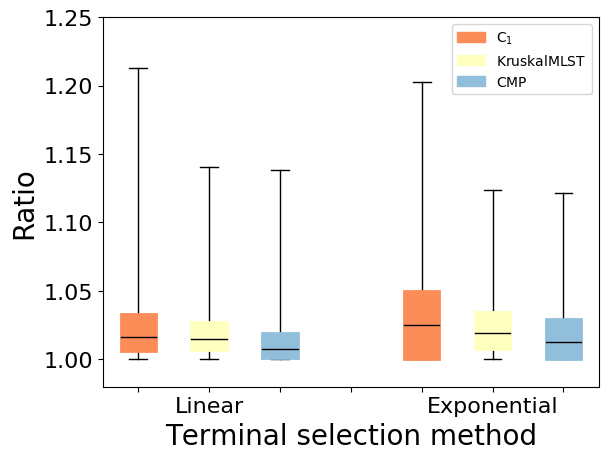}
    \end{subfigure}
    \caption{Performance of \QoSU~\cite{Charikar2004ToN}, \Kruskalu, and \CMP~\cite{MLST2018} on Barab{\'a}si--Albert graphs w.r.t.\
      $|V|$, $\ell$, and terminal selection method with proportional edge weights.}
    \label{BoxPlots_BA_uniform_cmp}
\end{minipage}
\end{figure}

We see that for Erd\H{o}s--R\'{e}nyi graphs, the number of vertices marginally increases the approximation ratio over time, while for the other generators this does not appear to be the case. Overall, no discernible trend occurs for the number of priorities regardless of the generator.  Interestingly, for randomly generated graphs, there appears to be no relation to the rate of decrease of terminal sets (i.e., linear vs. exponential) with the statistics of the approximation ratios.

\subsection{Approximation Ratio vs. Parameters -- Non-Proportional Edge Costs}

Here we consider the case non-proportional edge cost, in which we compare Algorithms \QoSNU and \Kruskalu.  The Composite algorithm of \cite{MLST2018} was not designed for non-proportional edge costs and so is not included here. Figures \ref{BoxPlots_ER_nonuniform}--\ref{BoxPlots_BA_nonuniform} show the approximation ratios vs. parameters for each of the random graph generators discussed above.

\begin{figure}[h!]
\begin{minipage}{\textwidth}
    \centering
    \begin{subfigure}[b]{0.30\textwidth}
        \includegraphics[width=\textwidth]{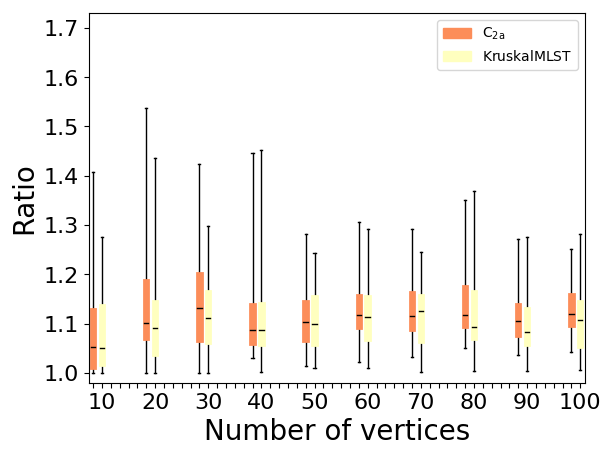}
    \end{subfigure}
    ~ 
    \begin{subfigure}[b]{0.30\textwidth}
        \includegraphics[width=\textwidth]{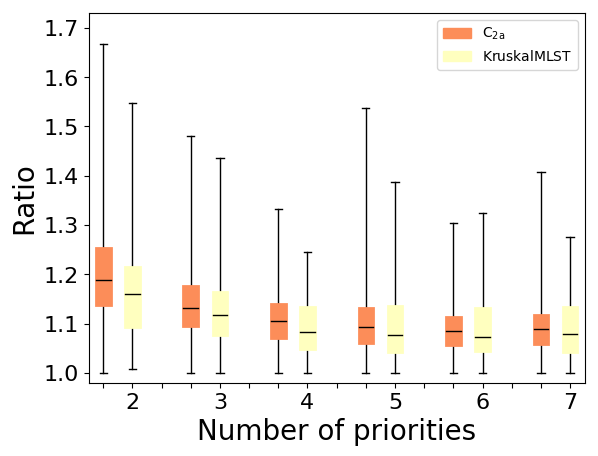}
    \end{subfigure}
    ~
    \begin{subfigure}[b]{0.30\textwidth}
        \includegraphics[width=\textwidth]{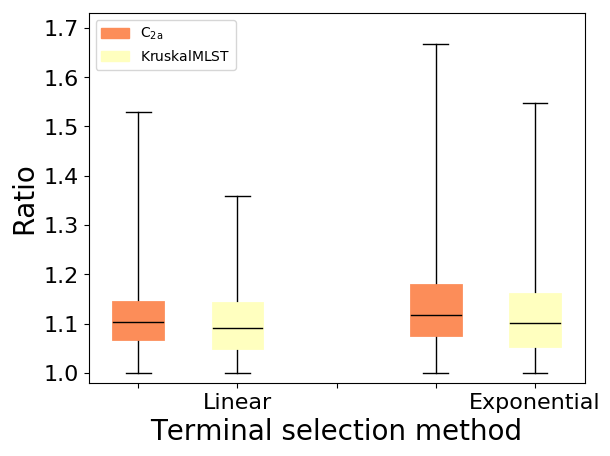}
    \end{subfigure}
    \caption{Performance of \QoSNU~\cite{Charikar2004ToN} and \Kruskalu w.r.t.\
      $|V|$, $\ell$, and terminal selection method with non-proportional edge weights on Erd\H{o}s--R{\'e}nyi graphs.}
    \label{BoxPlots_ER_nonuniform}
\end{minipage}
\end{figure}

\begin{figure}[h!]
\begin{minipage}{\textwidth}
    \centering
    \begin{subfigure}[b]{0.30\textwidth}
        \includegraphics[width=\textwidth]{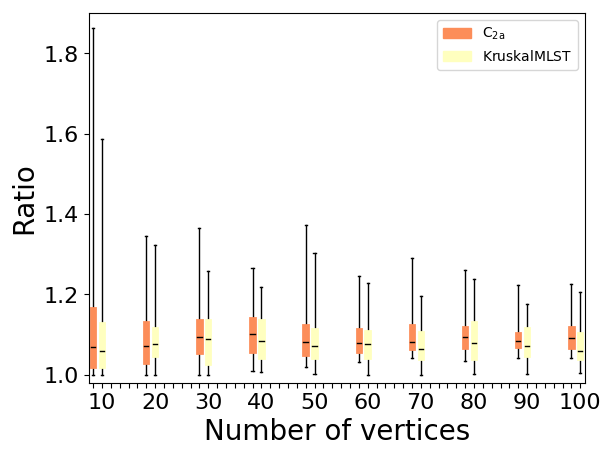}
    \end{subfigure}
    ~ 
    \begin{subfigure}[b]{0.30\textwidth}
        \includegraphics[width=\textwidth]{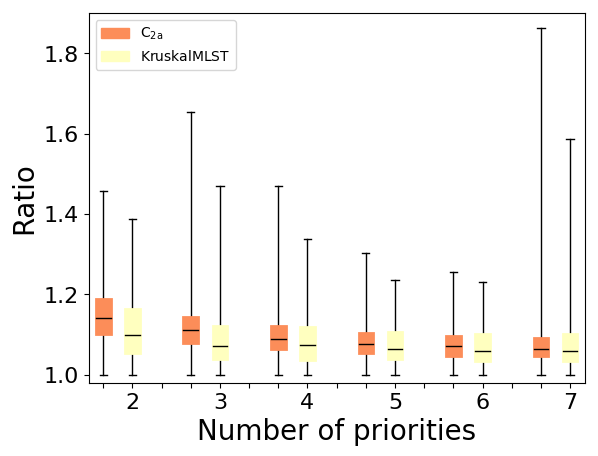}
    \end{subfigure}
    ~
    \begin{subfigure}[b]{0.30\textwidth}
        \includegraphics[width=\textwidth]{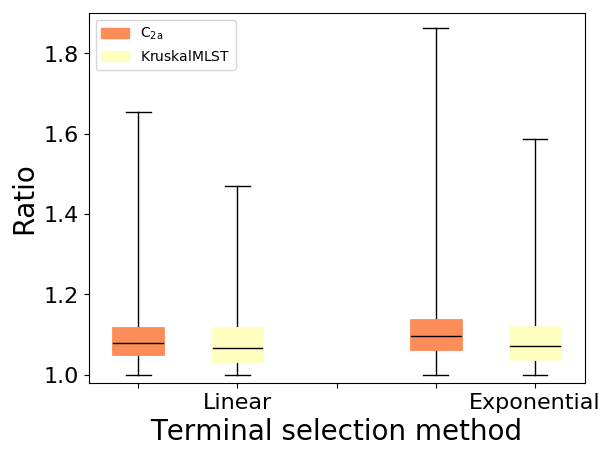}
    \end{subfigure}
    \caption{Performance of \QoSNU~\cite{Charikar2004ToN} and \Kruskalu w.r.t.\
      $|V|$, $\ell$, and terminal selection method with non-proportional edge weights on Watts--Strogatz graphs.}
    \label{BoxPlots_WS_nonuniform}
\end{minipage}
\end{figure}

\begin{figure}[h!]
\begin{minipage}{\textwidth}
    \centering
    \begin{subfigure}[b]{0.30\textwidth}
        \includegraphics[width=\textwidth]{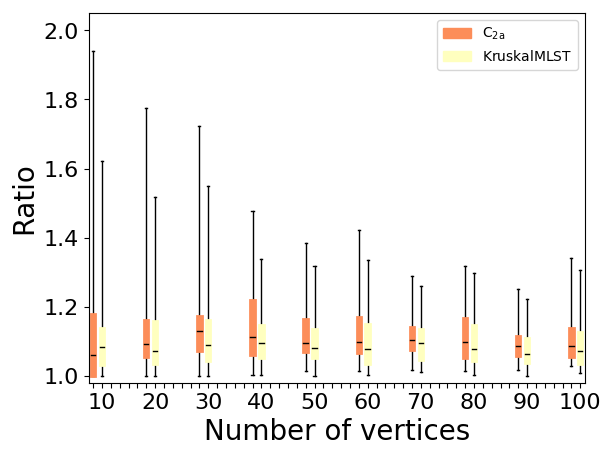}
    \end{subfigure}
    ~ 
    \begin{subfigure}[b]{0.30\textwidth}
        \includegraphics[width=\textwidth]{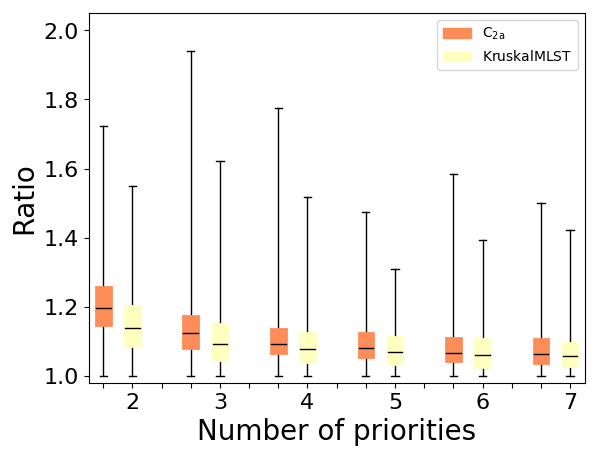}
    \end{subfigure}
    ~
    \begin{subfigure}[b]{0.30\textwidth}
        \includegraphics[width=\textwidth]{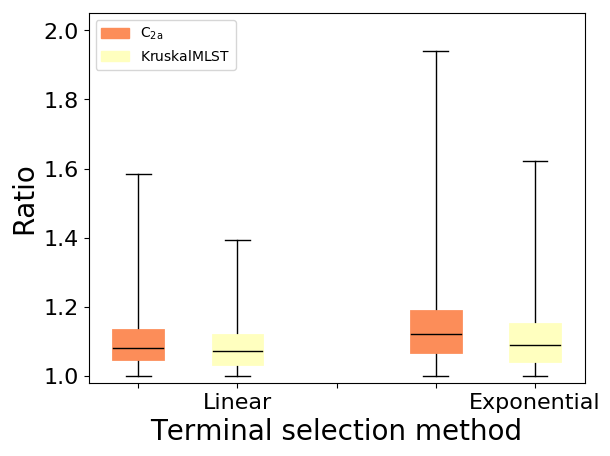}
    \end{subfigure}
    \caption{Performance of \QoSNU~\cite{Charikar2004ToN} and \Kruskalu w.r.t.\
      $|V|$, $\ell$ and terminal selection method with non-proportional edge weights on Barab{\'a}si--Albert graphs.}
    \label{BoxPlots_BA_nonuniform}
\end{minipage}
\end{figure}

In the non-proportional case, it is interesting that the approximation ratio appears to be little affected by any of the parameters, and even appears to decrease with respect to the number of priorities. It is unclear if this trend would continue for large number of priorities, but it is an interesting one nonetheless.  Of additional note is that \Kruskalu typically has less variance in its approximation ratio than the algorithms of Charikar et al.~\cite{Charikar2004ToN} in both the proportional and non-proportional case.

\subsection{Approximation Ratio vs. Parameters -- SteinLib Instances}

For the experiments on the SteinLib graphs \cite{KMV00}, we first extended two datasets (\href{http://steinlib.zib.de/showset.php?I080}{I080} and \href{http://steinlib.zib.de/showset.php?I160}{I160}) to have priorities via filtering or augmenting as described in Section \ref{section:experiment}. 
We provide the plots showing the Performance of \QoSU~\cite{Charikar2004ToN}, \Kruskalu, and \CMP~\cite{MLST2018} on \href{http://steinlib.zib.de/showset.php?I080}{I080} and \href{http://steinlib.zib.de/showset.php?I160}{I160} graphs w.r.t. $\ell$ with filtered priorities in Figure~\ref{BoxPlots_steinlib_filtered_cmp}, and for augmented priorities in Figure~\ref{BoxPlots_steinlib_augmented_cmp}. 

\begin{figure}[h!]
\begin{minipage}{\textwidth}
    \centering
    \begin{subfigure}[b]{0.30\textwidth}
        \includegraphics[width=\textwidth]{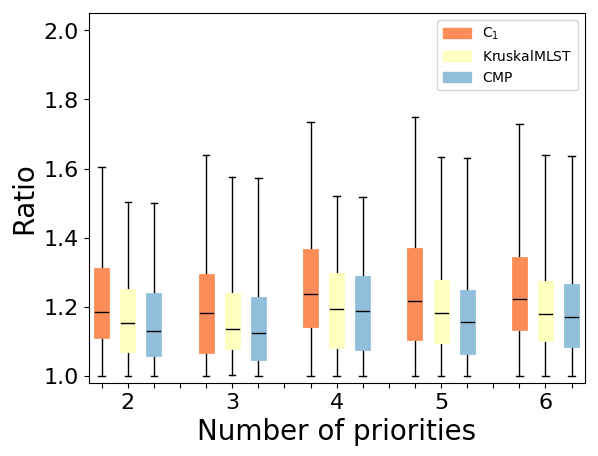}
    \end{subfigure}
    ~ 
    \begin{subfigure}[b]{0.30\textwidth}
        \includegraphics[width=\textwidth]{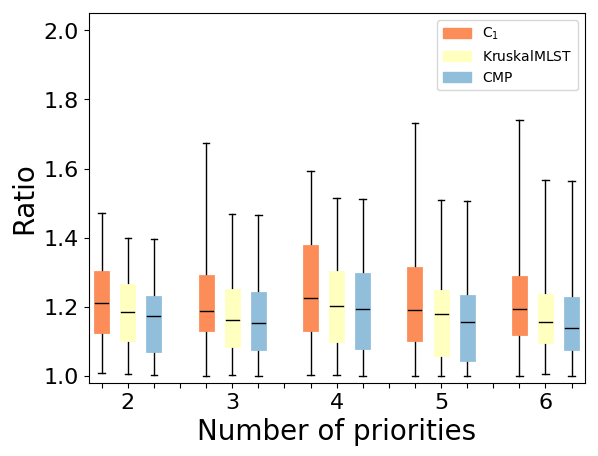}
    \end{subfigure}
    \caption{Performance of \QoSU~\cite{Charikar2004ToN}, \Kruskalu, and \CMP~\cite{MLST2018} on \href{http://steinlib.zib.de/showset.php?I080}{I080} and \href{http://steinlib.zib.de/showset.php?I160}{I160} graphs w.r.t.\
      $\ell$ with filtered priorities.}
    \label{BoxPlots_steinlib_filtered_cmp}
\end{minipage}
\end{figure}

\begin{figure}[h!]
\begin{minipage}{\textwidth}
    \centering
    \begin{subfigure}[b]{0.30\textwidth}
        \includegraphics[width=\textwidth]{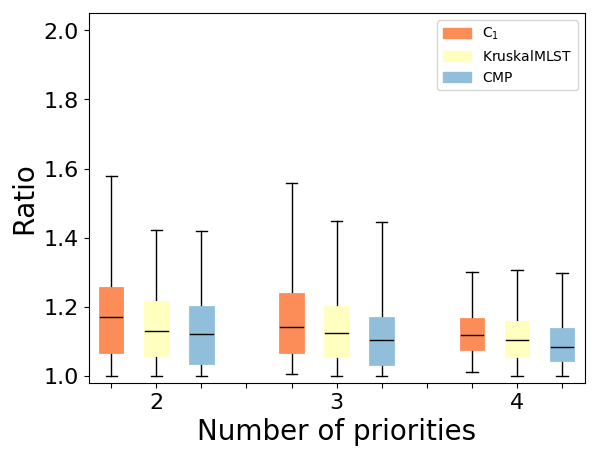}
    \end{subfigure}
    ~ 
    \begin{subfigure}[b]{0.30\textwidth}
        \includegraphics[width=\textwidth]{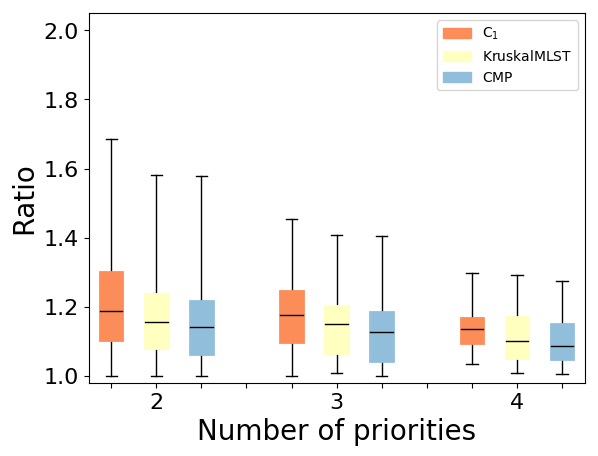}
    \end{subfigure}
    \caption{Performance of \QoSU~\cite{Charikar2004ToN}, \Kruskalu, and \CMP~\cite{MLST2018} on \href{http://steinlib.zib.de/showset.php?I080}{I080} and \href{http://steinlib.zib.de/showset.php?I160}{I160} graphs w.r.t.\
      $\ell$ with augmented priorities.}
    \label{BoxPlots_steinlib_augmented_cmp}
\end{minipage}
\end{figure}

\subsection{Runtime vs. Parameters -- Proportional Edge Costs}

Now we take a look at the affect of the parameters mentioned above on the average runtimes of the approximation algorithms in the case of proportional edge costs.  Figures \ref{BoxPlots_ER_uniform_time_heu}--\ref{BoxPlots_BA_uniform_time_heu} show the runtime of the algorithms \QoSNU, \Kruskalu, and \Kruskalnu versus $|V|$, $\ell$, and the terminal selection method.

\begin{figure}[h!]
\begin{minipage}{\textwidth}
    \centering
    \begin{subfigure}[b]{0.30\textwidth}
        \includegraphics[width=\textwidth]{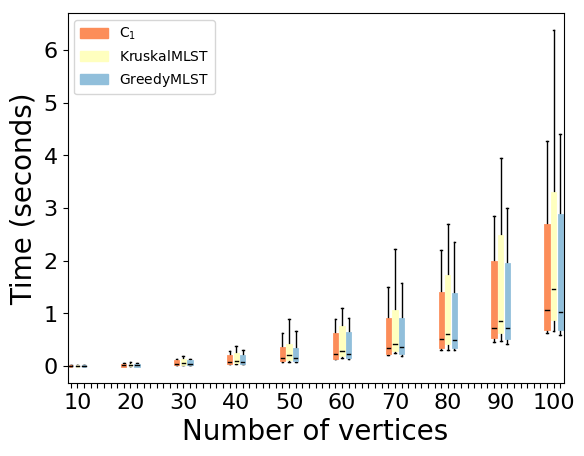}
    \end{subfigure}
    ~ 
    \begin{subfigure}[b]{0.30\textwidth}
        \includegraphics[width=\textwidth]{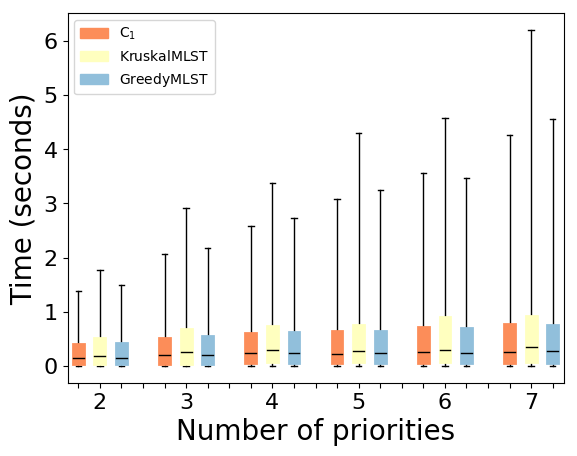}
    \end{subfigure}
    ~
    \begin{subfigure}[b]{0.30\textwidth}
        \includegraphics[width=\textwidth]{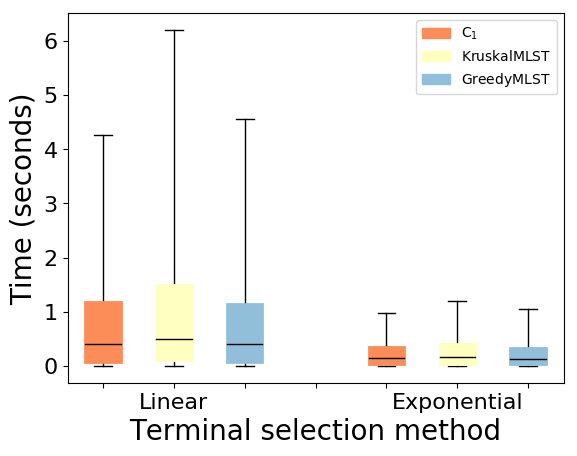}
    \end{subfigure}
    \caption{Experimental running times for computing approximation algorithm solutions w.r.t.\
      $|V|$, $\ell$, and terminal selection method with proportional edge weights on Erd\H{o}s--R{\'e}nyi graphs.}
    \label{BoxPlots_ER_uniform_time_heu}
\end{minipage}
\end{figure}

\begin{figure}[h!]
\begin{minipage}{\textwidth}
    \centering
    \begin{subfigure}[b]{0.30\textwidth}
        \includegraphics[width=\textwidth]{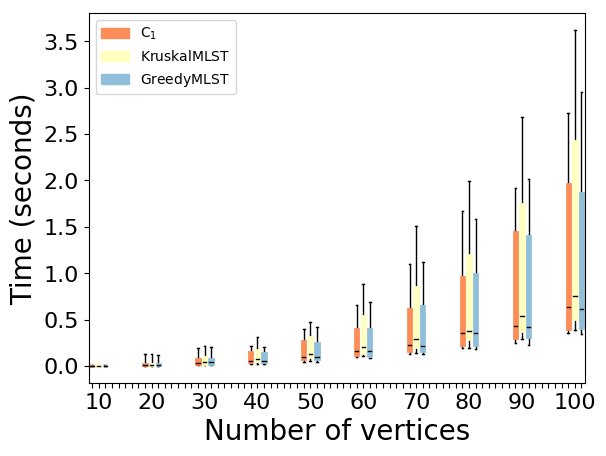}
    \end{subfigure}
    ~ 
    \begin{subfigure}[b]{0.30\textwidth}
        \includegraphics[width=\textwidth]{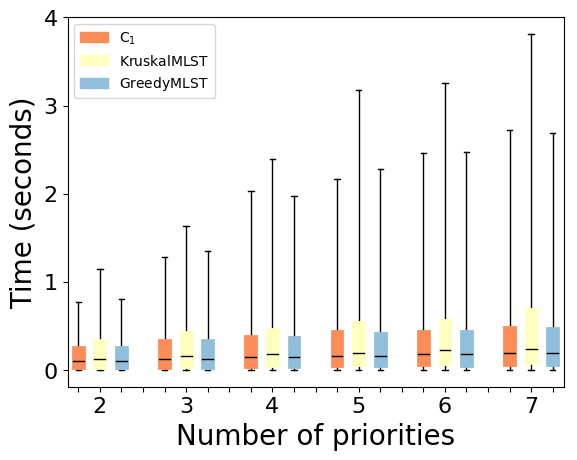}
    \end{subfigure}
    ~
    \begin{subfigure}[b]{0.30\textwidth}
        \includegraphics[width=\textwidth]{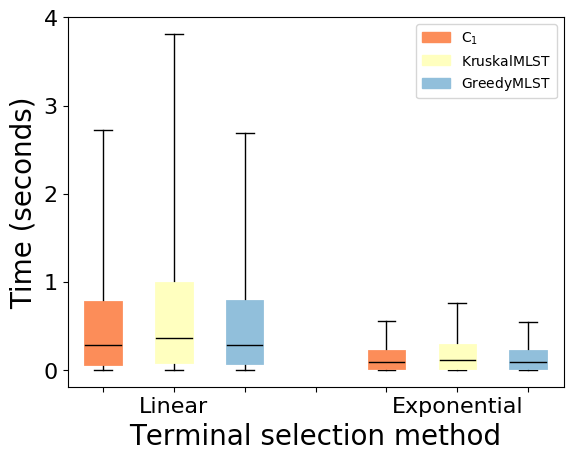}
    \end{subfigure}
    \caption{Experimental running times for computing approximation algorithm solutions w.r.t.\
      $|V|$, $\ell$, and terminal selection method with proportional edge weights on Watts--Strogatz graphs.}
    \label{BoxPlots_WS_uniform_time_heu}
\end{minipage}
\end{figure}

\begin{figure}[h!]
\begin{minipage}{\textwidth}
    \centering
    \begin{subfigure}[b]{0.30\textwidth}
        \includegraphics[width=\textwidth]{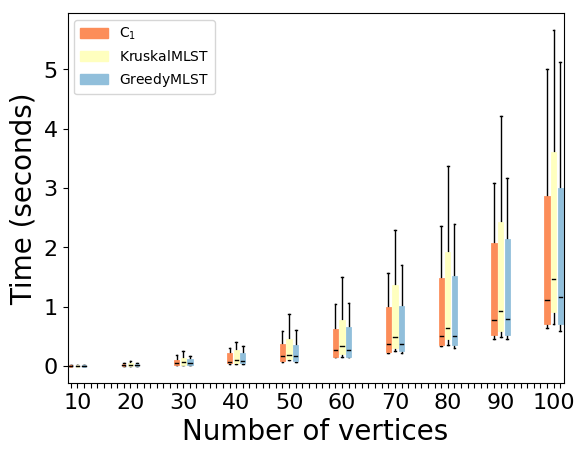}
    \end{subfigure}
    ~ 
    \begin{subfigure}[b]{0.30\textwidth}
        \includegraphics[width=\textwidth]{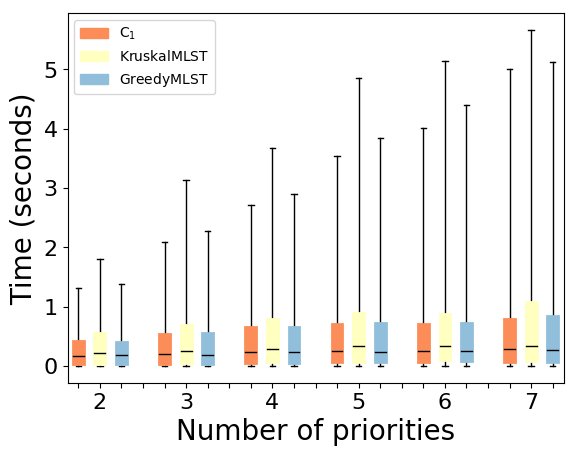}
    \end{subfigure}
    ~
    \begin{subfigure}[b]{0.30\textwidth}
        \includegraphics[width=\textwidth]{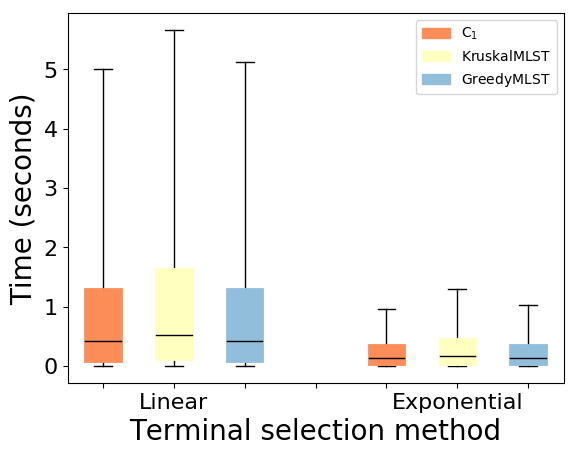}
    \end{subfigure}
    \caption{Experimental running times for computing approximation algorithm solutions w.r.t.\
      $|V|$, $\ell$, and terminal selection method with proportional edge weights on Barab{\'a}si--Albert graphs.}
    \label{BoxPlots_BA_uniform_time_heu}
\end{minipage}
\end{figure}

As is to be expected, on all generators, the average runtime increases as $|V|$ increases, as does the variance in the runtime.  Interestingly, average runtime does not appear to be much affected by the number of priorities, although the variance in runtime does substantially increase with $\ell$.  Runtime is lower for exponentially decreasing terminals, which makes sense given that in this case, the overall size of the terminal sets is smaller than in the linearly decreasing case.

\subsection{Runtime vs. Parameters -- Non-Proportional Edge Costs}

Now we take a look at the affect of the parameters mentioned above on the average runtimes of the approximation algorithm in the non-proportional case.  Figures \ref{BoxPlots_ER_nonuniform_time_heu}--\ref{BoxPlots_BA_nonuniform_time_heu} show the runtime of the algorithms \QoSNU, \Kruskalu, and \Kruskalnu versus $|V|$, $\ell$, and the terminal selection method.

\begin{figure}[h!]
\begin{minipage}{\textwidth}
    \centering
    \begin{subfigure}[b]{0.30\textwidth}
        \includegraphics[width=\textwidth]{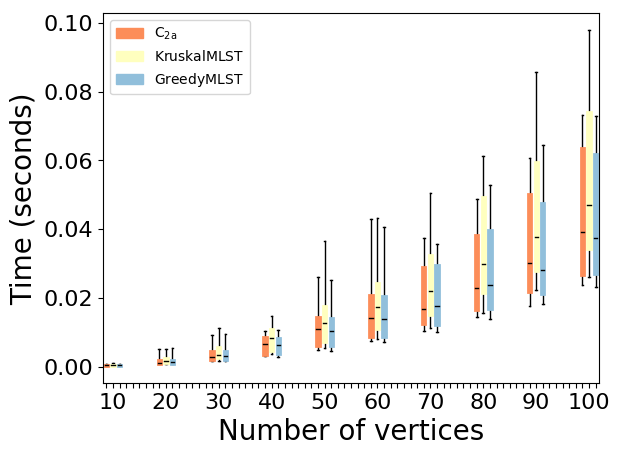}
    \end{subfigure}
    ~ 
    \begin{subfigure}[b]{0.30\textwidth}
        \includegraphics[width=\textwidth]{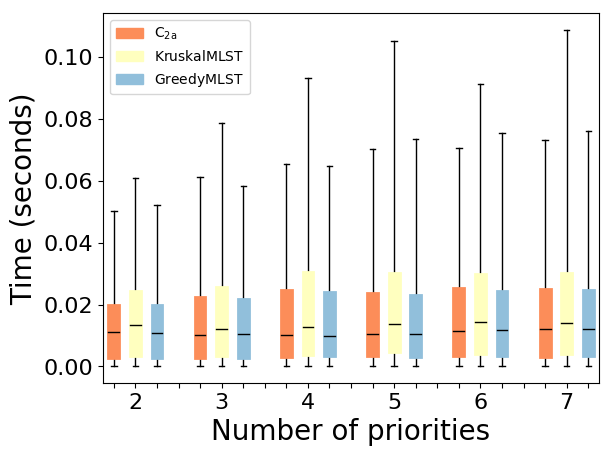}
    \end{subfigure}
    ~
    \begin{subfigure}[b]{0.30\textwidth}
        \includegraphics[width=\textwidth]{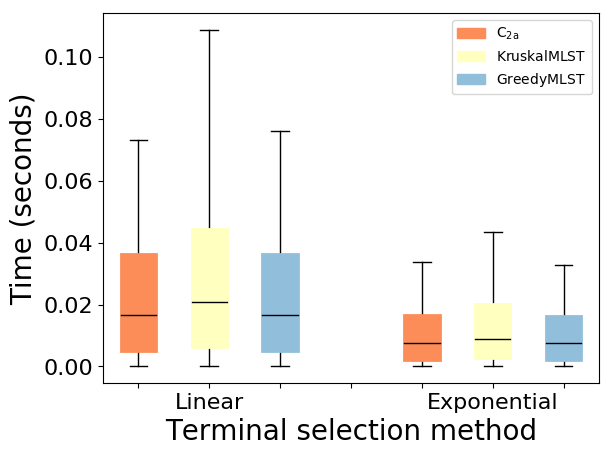}
    \end{subfigure}
    \caption{Experimental running times for computing approximation algorithm solutions w.r.t.\
      $|V|$, $\ell$, and terminal selection method with non-proportional edge weights on Erd\H{o}s--R{\'e}nyi graphs.}
    \label{BoxPlots_ER_nonuniform_time_heu}
\end{minipage}
\end{figure}

\begin{figure}[h!]
\begin{minipage}{\textwidth}
    \centering
    \begin{subfigure}[b]{0.30\textwidth}
        \includegraphics[width=\textwidth]{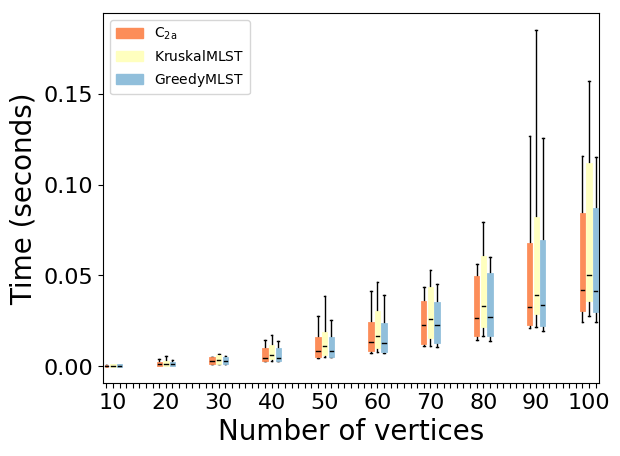}
    \end{subfigure}
    ~ 
    \begin{subfigure}[b]{0.30\textwidth}
        \includegraphics[width=\textwidth]{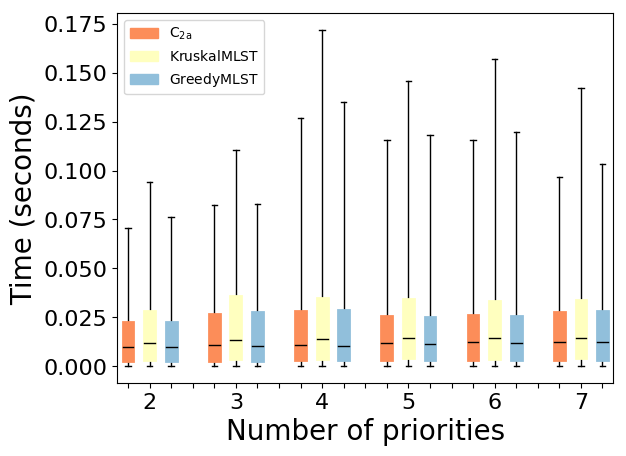}
    \end{subfigure}
    ~
    \begin{subfigure}[b]{0.30\textwidth}
        \includegraphics[width=\textwidth]{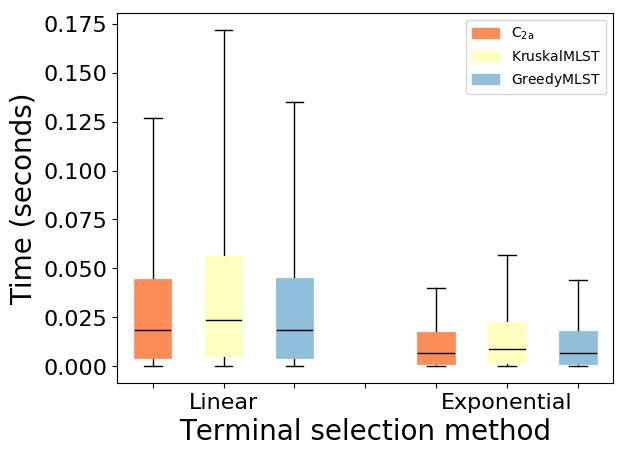}
    \end{subfigure}
    \caption{Experimental running times for computing approximation algorithm solutions w.r.t.\
      $|V|$, $\ell$, and terminal selection method with non-proportional edge weights on Watts--Strogatz graphs.}
    \label{BoxPlots_WS_nonuniform_time_heu}
\end{minipage}
\end{figure}

\begin{figure}[h!]
\begin{minipage}{\textwidth}
    \centering
    \begin{subfigure}[b]{0.30\textwidth}
        \includegraphics[width=\textwidth]{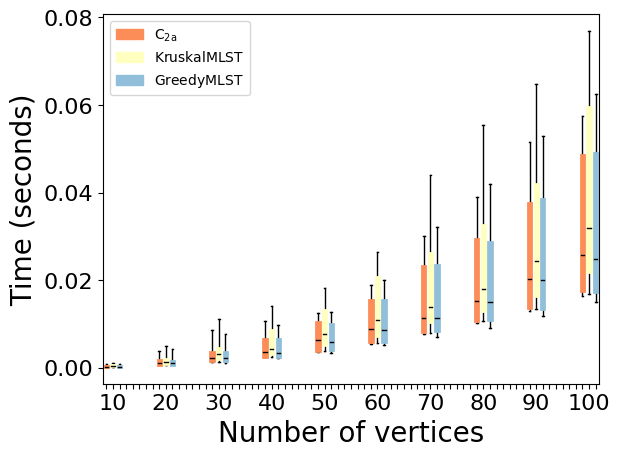}
    \end{subfigure}
    ~ 
    \begin{subfigure}[b]{0.30\textwidth}
        \includegraphics[width=\textwidth]{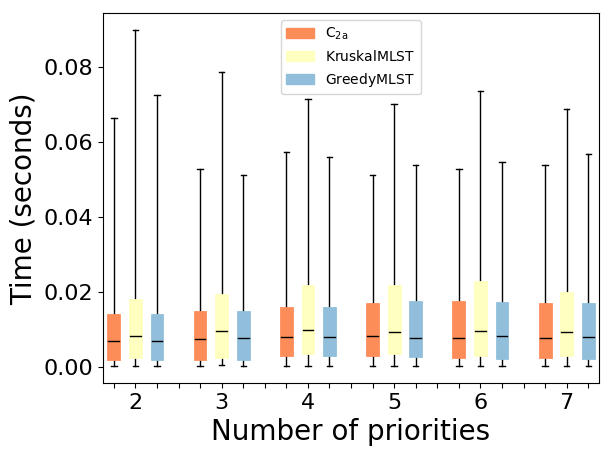}
    \end{subfigure}
    ~
    \begin{subfigure}[b]{0.30\textwidth}
        \includegraphics[width=\textwidth]{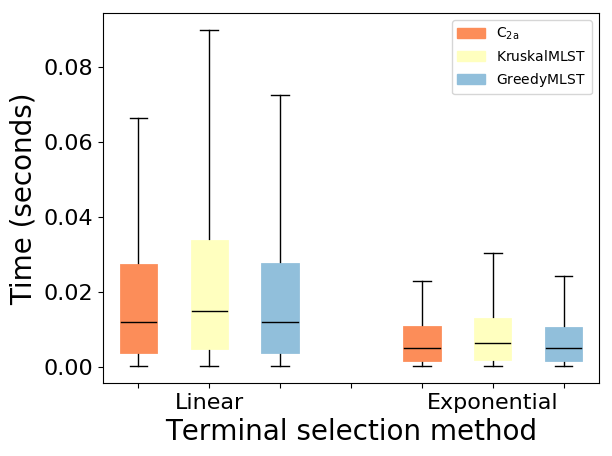}
    \end{subfigure}
    \caption{Experimental running times for computing approximation algorithm solutions w.r.t.\
      $|V|$, $\ell$, and terminal selection method with non-proportional edge weights on Barab{\'a}si--Albert graphs.}
    \label{BoxPlots_BA_nonuniform_time_heu}
\end{minipage}
\end{figure}

The trends are essentially the same as in the case of proportional edge costs; however, we note that the overall runtimes are almost two orders of magnitude smaller on average in the non-proportional trials run here.

\section{ILP Solver}

Without doubt, the most time consuming part of the experiments above was calculating the exact solutions of all MLST instances. For illustration, we show the runtime trends for the ILP solver with respect to $|V|$, $\ell$, and the terminal selection method for proportional edge costs in Figures \ref{BoxPlots_ER_uniform_time}--\ref{BoxPlots_BA_uniform_time} and for non-proportional edge costs in Figures \ref{BoxPlots_ER_nonuniform_time}--\ref{BoxPlots_BA_nonuniform_time} for all of the random graph generators.


\begin{figure}[h!]
\begin{minipage}{\textwidth}
    \centering
    \begin{subfigure}[b]{0.30\textwidth}
        \includegraphics[width=\textwidth]{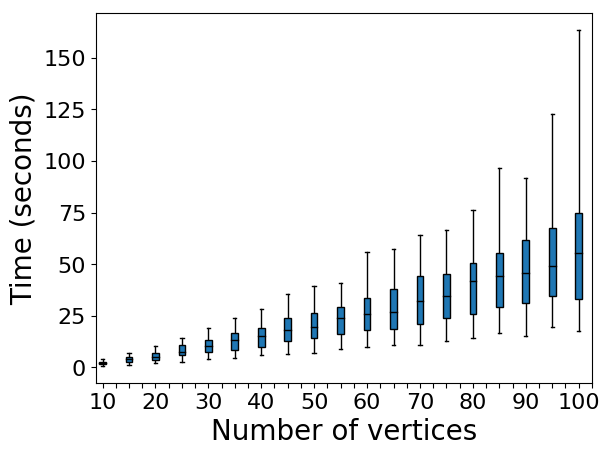}
    \end{subfigure}
    ~ 
    \begin{subfigure}[b]{0.30\textwidth}
        \includegraphics[width=\textwidth]{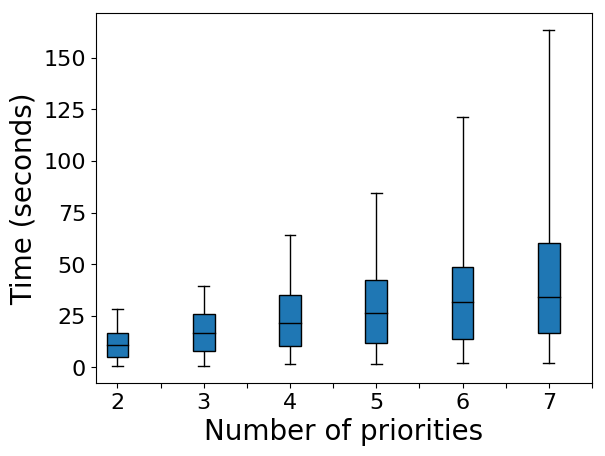}
    \end{subfigure}
    ~
    \begin{subfigure}[b]{0.30\textwidth}
        \includegraphics[width=\textwidth]{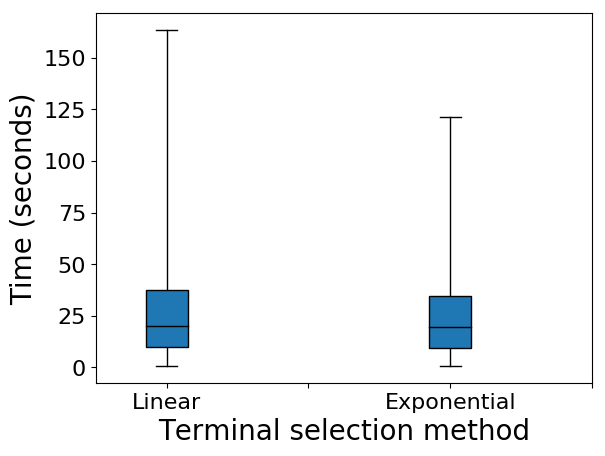}
    \end{subfigure}
    \caption{Experimental running times for computing exact solutions w.r.t.\
      $|V|$, $\ell$, and terminal selection method with proportional edge weights on Erd\H{o}s--R{\'e}nyi graphs.}
    \label{BoxPlots_ER_uniform_time}
\end{minipage}
\end{figure}

\begin{figure}[h!]
\begin{minipage}{\textwidth}
    \centering
    \begin{subfigure}[b]{0.30\textwidth}
        \includegraphics[width=\textwidth]{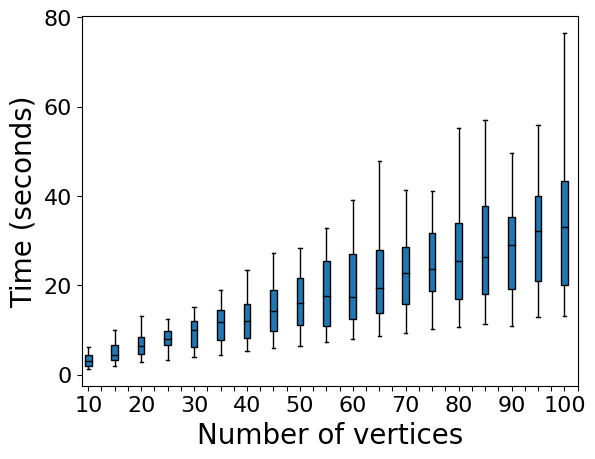}
    \end{subfigure}
    ~ 
    \begin{subfigure}[b]{0.30\textwidth}
        \includegraphics[width=\textwidth]{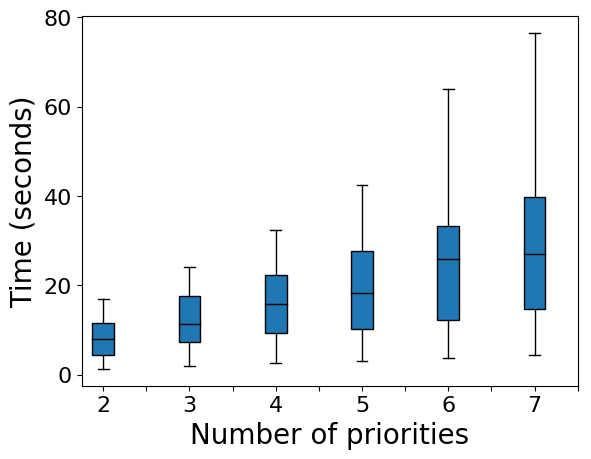}
    \end{subfigure}
    ~
    \begin{subfigure}[b]{0.30\textwidth}
        \includegraphics[width=\textwidth]{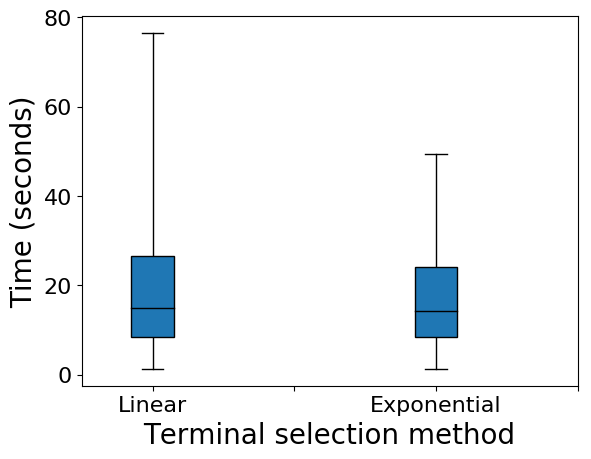}
    \end{subfigure}
    \caption{Experimental running times for computing exact solutions w.r.t.\
      $|V|$, $\ell$, and terminal selection method with proportional edge weights on Watts--Strogatz graphs.}
    \label{BoxPlots_WS_uniform_time}
\end{minipage}
\end{figure}

\begin{figure}[h!]
\begin{minipage}{\textwidth}
    \centering
    \begin{subfigure}[b]{0.30\textwidth}
        \includegraphics[width=\textwidth]{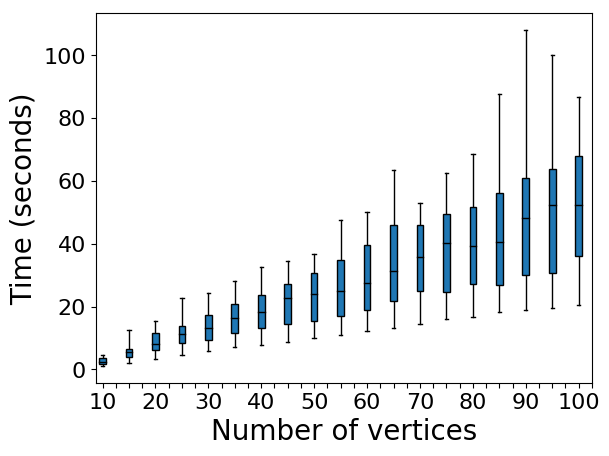}
    \end{subfigure}
    ~ 
    \begin{subfigure}[b]{0.30\textwidth}
        \includegraphics[width=\textwidth]{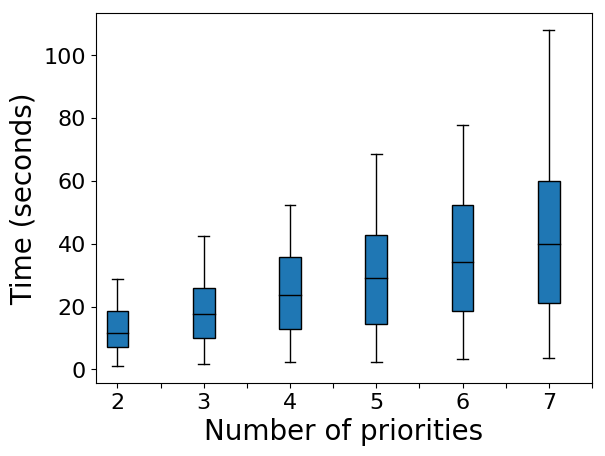}
    \end{subfigure}
    ~
    \begin{subfigure}[b]{0.30\textwidth}
        \includegraphics[width=\textwidth]{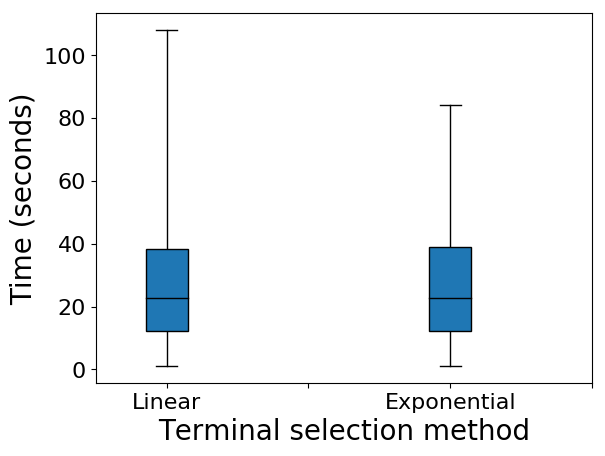}
    \end{subfigure}
    \caption{Experimental running times for computing exact solutions w.r.t.\
      $|V|$, $\ell$, and terminal selection method with proportional edge weights on Barab{\'a}si--Albert graphs.}
    \label{BoxPlots_BA_uniform_time}
\end{minipage}
\end{figure}

\begin{figure}[h!]
\begin{minipage}{\textwidth}
    \centering
    \begin{subfigure}[b]{0.30\textwidth}
        \includegraphics[width=\textwidth]{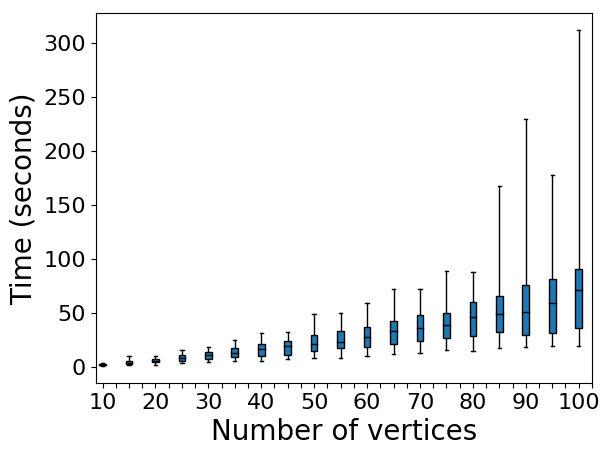}
    \end{subfigure}
    ~ 
    \begin{subfigure}[b]{0.30\textwidth}
        \includegraphics[width=\textwidth]{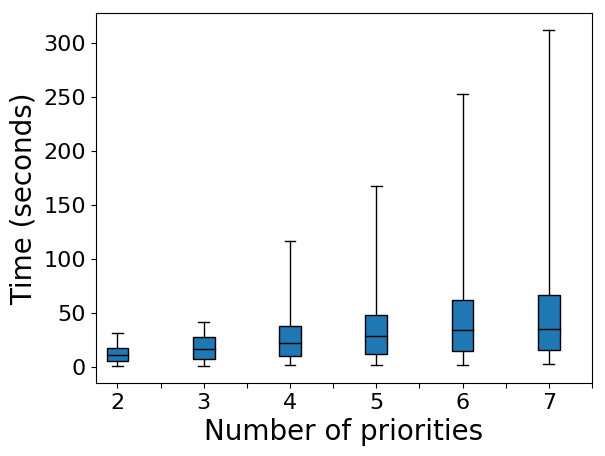}
    \end{subfigure}
    ~
    \begin{subfigure}[b]{0.30\textwidth}
        \includegraphics[width=\textwidth]{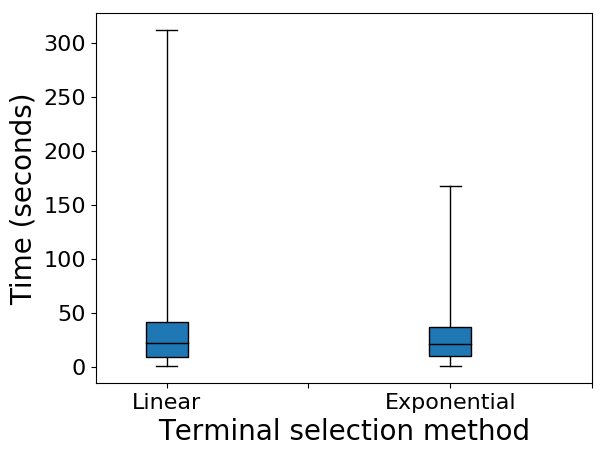}
    \end{subfigure}
    \caption{Experimental running times for computing exact solutions w.r.t.\
      $|V|$, $\ell$, and terminal selection method with non-proportional edge weights on Erd\H{o}s--R{\'e}nyi graphs.}
    \label{BoxPlots_ER_nonuniform_time}
\end{minipage}
\end{figure}

\begin{figure}[h!]
\begin{minipage}{\textwidth}
    \centering
    \begin{subfigure}[b]{0.30\textwidth}
        \includegraphics[width=\textwidth]{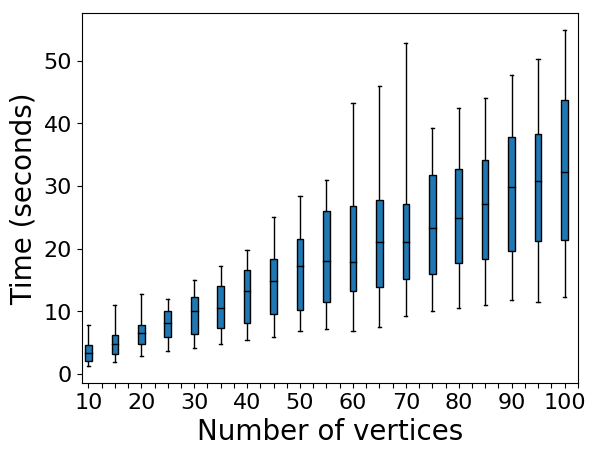}
    \end{subfigure}
    ~ 
    \begin{subfigure}[b]{0.30\textwidth}
        \includegraphics[width=\textwidth]{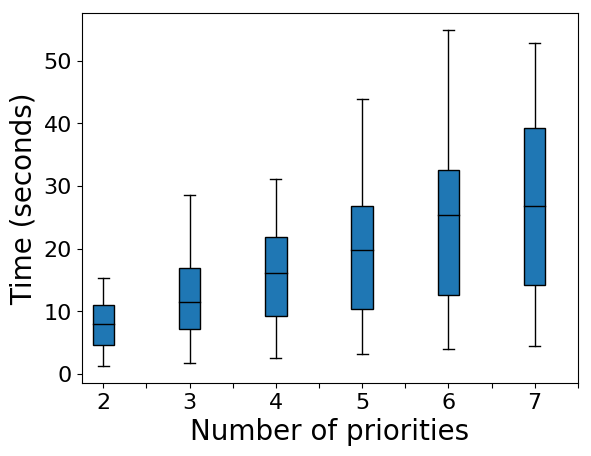}
    \end{subfigure}
    ~
    \begin{subfigure}[b]{0.30\textwidth}
        \includegraphics[width=\textwidth]{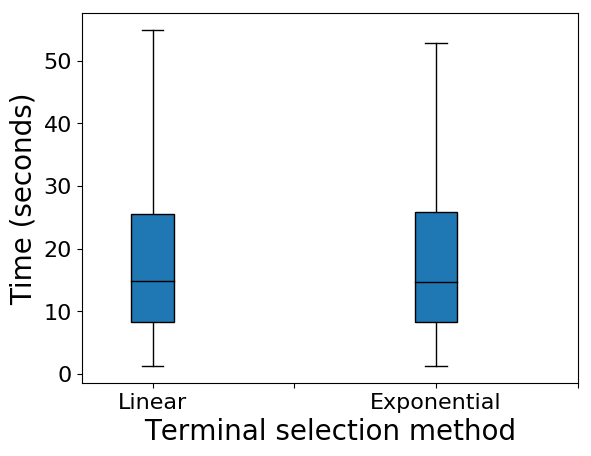}
    \end{subfigure}
    \caption{Experimental running times for computing exact solutions w.r.t.\
      $|V|$, $\ell$, and terminal selection method with non-proportional edge weights on Watts--Strogatz graphs.}
    \label{BoxPlots_WS_nonuniform_time}
\end{minipage}
\end{figure}

\begin{figure}[h!]
\begin{minipage}{\textwidth}
    \centering
    \begin{subfigure}[b]{0.30\textwidth}
        \includegraphics[width=\textwidth]{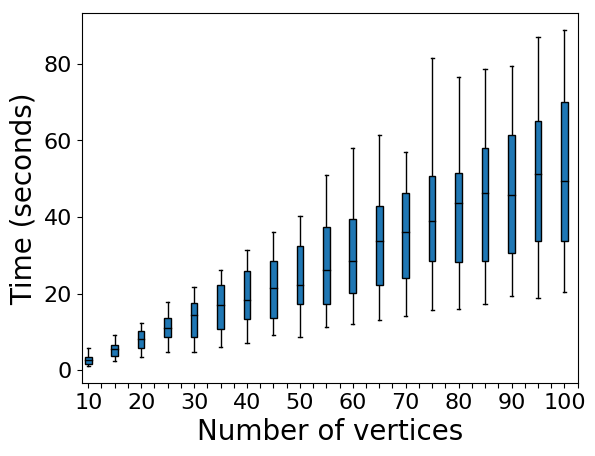}
    \end{subfigure}
    ~ 
    \begin{subfigure}[b]{0.30\textwidth}
        \includegraphics[width=\textwidth]{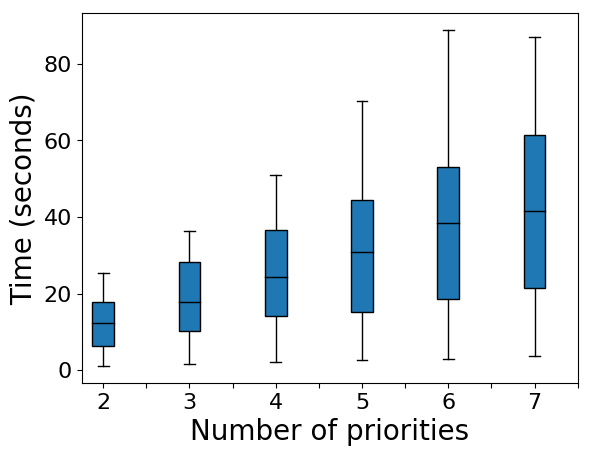}
    \end{subfigure}
    ~
    \begin{subfigure}[b]{0.30\textwidth}
        \includegraphics[width=\textwidth]{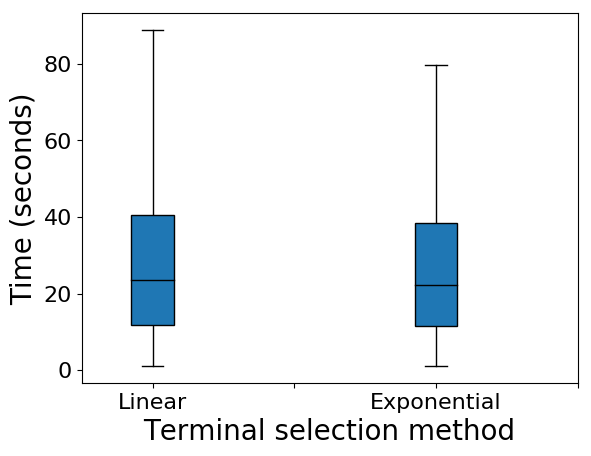}
    \end{subfigure}
    \caption{Experimental running times for computing exact solutions w.r.t.\
      $|V|$, $\ell$, and terminal selection method with non-proportional edge weights on Barab{\'a}si--Albert graphs.}
    \label{BoxPlots_BA_nonuniform_time}
\end{minipage}
\end{figure}

As expected, the running time of the ILP gets worse as $|V|$ and  $\ell$ increase. The running time of the ILP is worse for the linear terminal selection method, again likely because of the overall larger terminal set $T$. Note that the running time of the approximation algorithms are significantly faster than the running time of the exact algorithm. The exact algorithm takes a couple of minutes whereas the approximation algorithms take only a couple of seconds.

\end{document}